\documentclass{sig-alternate}
\newcommand{\nop}[1]{}
\makeatletter
\newif\if@restonecol
\makeatother

\usepackage[linesnumbered,ruled,vlined]{algorithm2e}
\usepackage{balance}
\usepackage{float}
\usepackage{graphicx}
\usepackage{epsfig}
\usepackage{float}
\usepackage{subfigure}
\usepackage{enumitem}
\usepackage{color}
\newtheorem{theorem}{Theorem}

\newtheorem{definition}{Definition}

\makeatletter
\def\@copyrightspace{\relax}
\makeatother

\begin{document}
%
% --- Author Metadata here ---
%\conferenceinfo{SAC'14}{March 24-28, 2014, Gyeongju, Korea.}
%\CopyrightYear{2014} % Allows default copyright year (2002) to be over-ridden - IF NEED BE.
%\crdata{978-1-4503-2469-4/14/03}  % Allows default copyright data (X-XXXXX-XX-X/XX/XX) to be over-ridden.
% --- End of Author Metadata ---

%\title{Comparing the Staples in Latent Factor Models}
\title{Buyer to Seller Recommendation under Constraints}
\nop{
}
%
% You need the command \numberofauthors to handle the 'placement
% and alignment' of the authors beneath the title.
%
% For aesthetic reasons, we recommend 'three authors at a time'
% i.e. three 'name/affiliation blocks' be placed beneath the title.
%
% NOTE: You are NOT restricted in how many 'rows' of
% "name/affiliations" may appear. We just ask that you restrict
% the number of 'columns' to three.
%
% Because of the available 'opening page real-estate'
% we ask you to refrain from putting more than six authors
% (two rows with three columns) beneath the article title.
% More than six makes the first-page appear very cluttered indeed.
%
% Use the \alignauthor commands to handle the names
% and affiliations for an 'aesthetic maximum' of six authors.
% Add names, affiliations, addresses for
% the seventh etc. author(s) as the argument for the
% \additionalauthors command.
% These 'additional authors' will be output/set for you
% without further effort on your part as the last section in
% the body of your article BEFORE References or any Appendices.

\numberofauthors{5} %  in this sample file, there are a *total*
% of EIGHT authors. SIX appear on the 'first-page' (for formatting
% reasons) and the remaining two appear in the \additionalauthors section.
%

\author{
% You can go ahead and credit any number of authors here,
% e.g. one 'row of three' or two rows (consisting of one row of three
% and a second row of one, two or three).
%
% The command \alignauthor (no curly braces needed) should
% precede each author name, affiliation/snail-mail address and
% e-mail address. Additionally, tag each line of
% affiliation/address with \affaddr, and tag the
% e-mail address with \email.
%
% 1st. author
%\affaddr{P.O. Box 1212}\\
\alignauthor
Cheng Chen\\
       \affaddr{University of Victoria}\\
       \affaddr{Victoria, Canada}\\
       \email{cchenv@uvic.ca}
% 2nd. author
\alignauthor
Lan Zheng\\
       \affaddr{University of Victoria}\\
       \affaddr{Victoria, Canada}\\
       \email{lanzheng@uvic.ca}
% 3rd. author
\alignauthor Venkatesh Srinivasan\\
       \affaddr{University of Victoria}\\
       \affaddr{Victoria, Canada}\\
       \email{venkat@cs.uvic.ca}
\and  % use '\and' if you need 'another row' of author names
% 4th. author
\alignauthor Alex Thomo\\
       \affaddr{University of Victoria}\\
       \affaddr{Victoria, Canada}\\
       \email{thomo@cs.uvic.ca}
% 5th. author
\alignauthor Kui Wu\\
       \affaddr{University of Victoria}\\
       \affaddr{Victoria, Canada}\\
       \email{wkui@uvic.ca}
% 6th. author
\alignauthor Anthony Sukow\\
       \affaddr{Terapeak}\\
       \affaddr{Victoria, Canada}\\
       \email{anthony@terapeak.com}
}

% There's nothing stopping you putting the seventh, eighth, etc.
% author on the opening page (as the 'third row') but we ask,
% for aesthetic reasons that you place these 'additional authors'
% in the \additional authors block, viz.

% Just remember to make sure that the TOTAL number of authors
% is the number that will appear on the first page PLUS the
% number that will appear in the \additionalauthors section.

\maketitle
\begin{abstract}

The majority of recommender systems are designed to recommend items 
(such as movies and products) to users. We focus on the problem of recommending buyers to sellers which comes with new challenges:
(1)~constraints on the number of recommendations buyers are part of before they become overwhelmed, 
(2)~constraints on the number of recommendations sellers receive within their budget, and
(3)~constraints on the set of buyers that sellers want to receive 
(e.g., no more than two people from the same household).
We propose the following critical problems of recommending buyers to sellers:  
Constrained Recommendation (C-REC) capturing the first two challenges, and 
Conflict-Aware Constrained Recommendation (CAC-REC) capturing all three challenges at the same time. 
We show that C-REC can be modeled using linear programming and can be efficiently solved using modern solvers. 
On the other hand, we show that CAC-REC is NP-hard. 
We propose two approximate algorithms to solve CAC-REC and show that they achieve close to optimal solutions via comprehensive experiments using real-world datasets.

%The majority of recommender systems in E-commerce sites are designed to make recommendations of items to users. 
%Better item recommendations can improve the satisfaction of buyers and boost sales of sellers on the platform. 
%However, recommending buyers to sellers is often neglected. 
%This problem comes with a new set of challenges, such as the limited number sellers a buyer can be reasonably recommended to before he becomes overwhelmed with sellers requests, and constraints that sellers can put on the number and type of buyers threcommended to them 

%; e.g., a high valuable buyer who purchased many items might not be referred to a top seller of abundant products which can satisfy the buyer's appetite. In order to maximize potential overall profit, it is desirable for the service provider to explicitly recommend buyers to a specific set of sellers. We call this problem as recommending buyers to sellers (RBS).

%This paper studies two critical problems of RBS, conflict oblivious recommendation (C-REC) and conflict aware recommendation (CAC-REC). The conflict between buyers represents the need of diversifying buyers for sellers. While C-REC can be effectively solved using linear programming, CAC-REC is NP-hard. Two approximate algorithms are proposed to solve CAC-REC and are shown to achieve close to optimal solutions via comprehensive experiments using a real-world dataset.  
\end{abstract}



%% A category with the (minimum) three required fields
%\category{H.3}{Information Storage and Retrieval}{Information Search and Retrieval}[Information filtering]
%%A category including the fourth, optional field follows...
%\category{D.2.8}{Software Engineering}{Metrics}[complexity measures, performance measures]
%
%%\terms{Data mining}
%
%\keywords{Recommender systems, latent factor models, evaluation}

\section{Introduction}
We study the problem of recommending buyers to sellers (RBS) in an online e-commerce site, such as eBay. 
This problem comes with its own set of challenges, which are quite different from those of recommending items to users.
The first challenge is that we do not want to recommend the top buyers (in terms of their buying history) 
to all the sellers because those top buyers
will be inundated by thousands of advertisement messages from eager sellers trying to reach out to them. 
Therefore we need constraints in place that limit the number of sellers each buyer can be recommended to.
The second challenge is that sellers are normally under budget limitations, 
thus do not want more buyers recommended to them than what they can pay for. 
The third challenge is that we need to recommend 
top buyers to top sellers, 
middle-tier buyers to middle-tier sellers, and so on, 
as this maximizes the chance 
of making both the buyers and sellers happy; typically, top buyers buy a lot of merchandise, and it is the top 
sellers that have the best chance of satisfying their appetite.
In general we would like to avoid top buyers being recommended to less than stellar sellers
unless the latter have already been saturated with recommended buyers. 

We do not normally have these challenges when recommending items (e.g., movies) to users (RIU). 
For example there is absolutely no problem recommending the best movies to all the users of an online site;
on the contrary, it is widely applied. %(e.g. Netflix recommendation of ``Mad Men'' in March 2014)
Also, there is no problem of recommending too many movies to one person as long as they are sorted by the rating prediction. 
And, finally, movie recommendation is a very ``democratic'' process; 
top movies are recommended all the time to less than stellar users, i.e., those that do not watch and rate many movies (cold-start users). 
%In fact it is the recommendations for great movies that have a real chance of wakening up a cold-start user to interact more with an online site.  

From all the above, it is clear that RBS is very different from RIU, 
and therefore, we need a new formalization and new algorithms to address RBS. To the best of our knowledge, we are the first to directly approach this important problem. This is because classical advertising is done in ``broadcast'' mode via TV commercials and sponsored links in search results, which do not have the aforementioned challenges. Another way of advertising is via targeted campaigns. This is typically done (or initiated) by a single merchant company and involves identifying the most promising customers to target for its own products. Again, such a process is different from RBS, where there are many sellers competing for the buyers. 

In this paper, 
we introduce two central problems, Constrained Recommendation (C-REC) and Conflict-Aware Constrained Recommendation (CAC-REC), to address different RBS scenarios. We view RBS from the point of view of an online company, like eBay. Such a company typically associates a profit weight for each recommendation of a buyer to a seller. 
For a given seller, the more highly ranked the buyer, the more profitable the recommendation. 

In {\sc C-REC}, we assume that each buyer and seller has a specific ``capacity.'' That is, a buyer cannot be recommended to more sellers than his capacity, 
and a seller cannot be recommended more buyers than he has budget (capacity) for. Then the goal is to generate recommendations maximizing the total profitability (sum of recommendation weights), while respecting the capacity constraints of the buyers and sellers. 

\nop{In P2, we extend P1 by assuming that buyers and sellers are not disjoint sets. 
This is often the case in communities of sellers, where
they buy/sell from each other to replenish/free-up their repositories,
and the problem is to not overwhelm them by recommending too many connections. Again here, we want the top sellers to get in touch with other top sellers in order to maximize the chance they satisfy each other. }

In CAC-REC, we extend C-REC in a natural way. Namely, we assume the buyers could be ``in confict'' with other buyers. For example, a buyer might use the same IP address or browser cookie as another buyer, which would indicate that these two buyers share the same computer and probably live together in the same household. A seller might not want to be recommended buyers that are in conflict with each other, as for example in the case when the seller sells household items such as TVs, washers, etc. 

We show that C-REC can be modeled using linear programming and can be efficiently solved using modern solvers. 
On the other hand, we show that CAC-REC is NP-hard. 
We propose two approximate algorithms to solve CAC-REC and show that they achieve close to optimal solutions via comprehensive experiments using real-world datasets.

More specifically, our contributions are as follows:
\begin{enumerate}
\item
We initiate the study of natural problems arising in RBS systems: How do we recommend high-value buyers to sellers under various constraints?
\item
In the presence of capacity constraints (C-REC), we model the problem as integer linear programming and show that the problem can be optimally solved using LP solvers (i.e., the LP solutions are integral). 
\item
For the case of capacity and conflict constraints (CAC-REC), we prove that the problem is NP-hard and model it as semidefinite programming and integer linear programming. We present a greedy algorithm that is scalable and close to optimal.
\item
We provide an extensive experimental evaluation on real-world datasets validating the claims of scalability and optimality made above.
\end{enumerate}

%\newpage
\section{Related Work}
%ALEX editing

RBS is related to the task of constraint-based recommendation
(cf. \cite{DBLP:conf/ACMicec/FelfernigB08,DBLP:journals/umuai/ZankerJ09,DBLP:conf/cikm/KarimzadehganZ09,karimzadehgan2012integer,taylor2008optimal,DBLP:conf/recsys/XieLW10,DBLP:journals/tois/ParameswaranVG11}). 

The first two works consider constraints on item features (attributes), 
thus they do not study the same type of constraints as we do. 
The rest of works are more closely related to our first, C-REC, problem. 

Karimzadehgan et. al in \cite{karimzadehgan2008multi,DBLP:conf/cikm/KarimzadehganZ09,karimzadehgan2012integer} study the problem of optimizing the review assignments of scientific papers. Similarly to C-REC, they employ constraints on the quota of papers each reviewer is assigned. However, differently from C-REC, in their optimization setup, matching of reviewers with a paper is done based on matching of multiple aspects of expertise.
Taylor in \cite{taylor2008optimal} also considers the paper review assignment problem. The difference from C-REC is that it does not consider an ordering on the reviewers and papers. 

Xie, Lakshmanan, and Wood in \cite{DBLP:conf/recsys/XieLW10} study the problem of composite recommendations, where each recommendation comprises a set of items. 
They also consider constraints including the number of items that can be recommended to a user. Their objective, however, is to minimize the cost of a recommended set of items when each item has a price to be paid. 

Parameswaran, Venetis, and Garcia-Molina in \cite{DBLP:journals/tois/ParameswaranVG11} study the problem of course recommendations with course requirement constraints. Similarly as \cite{DBLP:conf/recsys/XieLW10}, the goal of \cite{DBLP:journals/tois/ParameswaranVG11} is to come up with set recommendations. However, the challenge they address is the modeling of complex academic requirements (e.g., take $2$ out of a set of $5$ math courses to meet the degree requirement). Such constraints are different from those that we consider.  

To the best of our knowledge, the second problem we consider, {\em CAC-REC}, is completely new.

\section{Constrained Recommendation \\(C-REC)} %\linebreak
\label{P1}

We first study the recommendation problem without considering conflict between buyers. To describe the problem formally, we model the buyer-seller network as a undirected, bipartite graph (Figure~\ref{fig:p1_figure}) $G=\left<(B,S), E, W\right>$, where $B=\{b_1,b_2,\ldots, b_m\}$ denotes the list of buyers, $S=\{s_1,s_2,\ldots,s_n\}$ the list of sellers, $E \subseteq B \times S$ the set of edges, and $W: E \rightarrow \mathbb{R}^+$ the weights of the edges. 

\begin{figure}[H]
\centering
\includegraphics[scale=0.4]{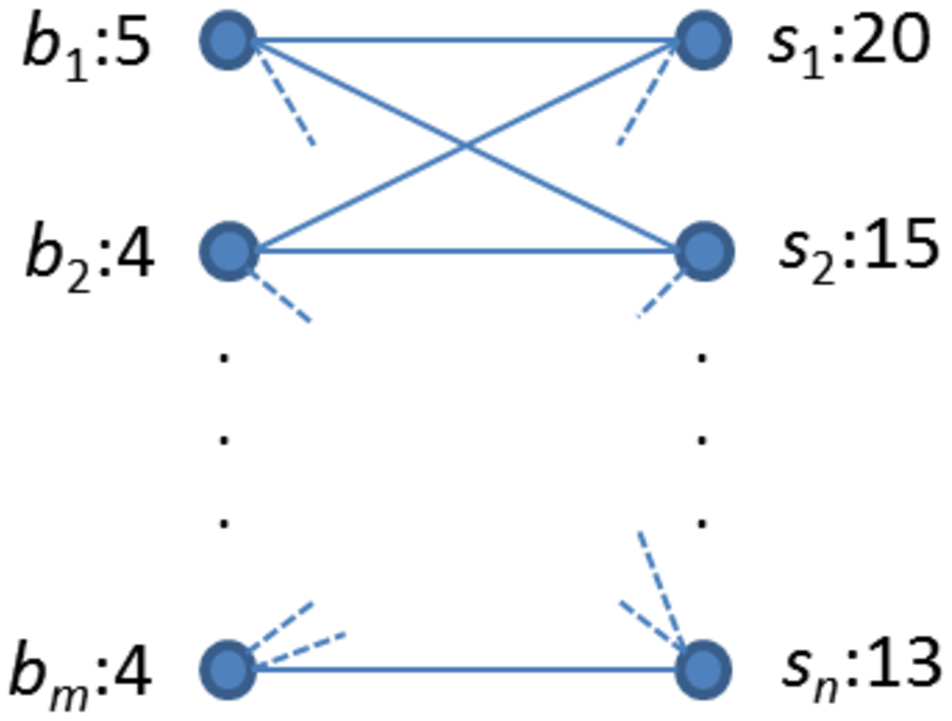}
\caption{C-REC. The numbers are degree constraints.}
\label{fig:p1_figure}
\end{figure}

For convenience, we slightly abuse the notation and use an $mn$-dimensional vector to denote $E= [e_{ij}]$, with $e_{ij}=1$ indicating that there is an edge between buyer $i$ and seller $j$ and $e_{ij}=0$ otherwise. Similarly, we use a vector to denote $W=[w_{ij}]$. If $e_{ij}=0$, then $w_{ij}=0$. When the subscripts are hard to read, as in the next section, we also use $W(i,j)$ to denote $w_{ij}$.     

The weight of the edge between a buyer $i \in B$ and a seller $j\in S$, $w_{ij}$, reflects the profitability value if the buyer is recommended to the seller. In practice, we may pre-process the bipartite graph based on various business models. For instance, we may order the buyers or sellers based on money spent and earned, recency of buys and sells, etc. In addition, we may constrain the edges from a buyer to a ranked range of sellers, so that the buyer is not recommended to sellers well outside of her ``tier".   

It is often desirable that only a certain number of buyers is recommended to a seller, since a seller may be overwhelmed otherwise or because a seller needs to pay for the recommendation. Similarly, it is not reasonable to recommend a buyer to a large number of sellers, as the buyer might be annoyed later by many unsolicited requests. To avoid the problem, a good recommendation system should allow us to constrain the number of recommendations associated with individual buyers and sellers. In other words, we need to put degree constraints $D: B\cup S \rightarrow \mathbb{N}$ in the bipartite graph. We represent $D$ with a $(m+n)$-dimensional column vector $D = [D(i)]^T$.  

Denote $X=[x_{ij}]^T$ as the $mn$-dimensional column vector of $0$-$1$ variables, with $x_{ij}=1$ indicating buyer $i$ is recommended to seller $j$ and $x_{ij}=0$ otherwise. The constrained recommendation ({\sc C-REC}) problem is to \textit{find the set of recommendations such that the total profit of the recommendation is maximized under the degree constraints}, i.e.,   
\begin{equation}\label{prob:CoRec}
\begin{aligned}
&  \max_{X}
& & WX  \\
& \textit{s.t.}
& & \mathbb{A} X (i)  \leq D(i), \forall i, 1 \leq i \leq m+n \\  %\preceq
&&&  x_{ij} \in \{0,1\}, \forall i,j, 1 \leq i \leq m, 1 \leq j \leq n,  % \mathbb{A}_2 X \leq \mathbb{I}_{mn}
\end{aligned}
\end{equation}
where matrix $\mathbb{A}$ is an $(m+n) \times mn$ matrix defined by 
\begin{equation}\label{eqt:A}
\begin{aligned}
\underbrace{
  \begin{bmatrix}
    [e_{11}, \ldots, e_{1n}] &    &           \\
       & [e_{21}, \ldots, e_{2n}] &           \\
	   & \ddots  &[e_{m1}, \ldots, e_{mn}]    \\
       [e_{11}, 0, \ldots, 0] & \ldots & [e_{m1}, 0, \ldots,0]          \\
       \ddots & \ddots& \ddots \\
       [0,\ldots, 0, e_{1n}] & \ldots & [0,\ldots, 0, e_{mn}] \\
  \end{bmatrix}.
	}_{(m+n)\times mn}
\end{aligned}
\end{equation}
The degree constraints are given by $\mathbb{A} X (i)  \leq D(i)$, where  $\mathbb{A} X (i)$ denotes the $i$-th element in (vector) $\mathbb{A} X$ and $D(i)$ the $i$-th element in $D$.

\nop{\noindent where $\mathbb{I}_{n\times n}$ is an $n \times n$ identity matrix, and the $0$-$1$ constraints on variable $x_{ij}$ are given by $\mathbb{A}_2 X \preceq \mathbb{I}_{mn}$, where $\mathbb{A}_2$ is an $mn\times mn$ identity matrix and $\mathbb{I}_{mn}$ is an $mn$-dimensional all $1$ vector. } 

From the above formulation, it is easy to see that the {\sc C-REC} problem could be reduced to the transportation problem in operations research~\cite{Amini,Rebman197411, Schrijver:1986:TLI:17634}, and as such we obtain an LP problem that can be solved efficiently by the modern solvers. Furthermore, it can be shown the LP solution is also integral.  

\section{Conflict-Aware Constrained \\Recommendation (CAC-REC)}
\label{P3}
We now consider a natural generalization of problem C-REC. In some scenarios, sellers will prefer a {\em diverse} list of recommended buyers that avoids
certain redundancies. For example, a seller might prefer that their list does not include more than one recommended buyer from each household. Advertising to more than one potential buyer in a household for a given merchandise, in most cases, is unnecessary. We will represent presence of such dependencies or conflicts between two buyers using (unweighted) conflict edges (Figure~\ref{fig:p3_figure}).  We call this recommendation problem conflict-aware constrained recommendation ({\sc CAC-REC}). The goal is to \textit{compute a maximum weight subgraph satisfying the degree constraints as in C-REC with the additional requirement that the number of conflict edges within a list of buyers recommended to any particular seller is smaller than a threshold}. 

\begin{figure}[H]
\centering
\includegraphics[scale=0.4]{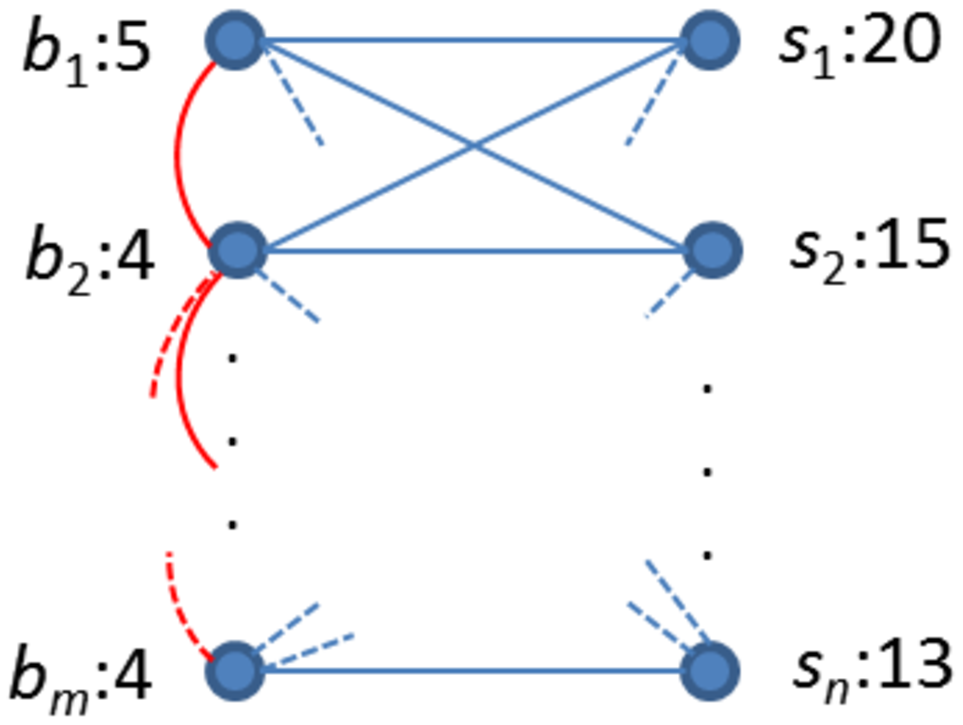}
\caption{CAC-REC. The numbers are degree constraints and the red edges represent conflicts.}
\label{fig:p3_figure}
\end{figure}

To describe {\sc CAC-REC} more precisely, the input to {\sc CAC-REC} consists of the following information:
\begin{enumerate}
\item
An undirected, weighted graph $G=\left<(B,S), E \cup C, W\right>$ with $E \subseteq B \times S$, $C \subseteq B \times B$ and weights $W: E \rightarrow \mathbb{R}^+$;
\item
Degree constraints $D: B\cup S \rightarrow \mathbb{N}$;
\item
A conflict threshold $t$.
\end{enumerate}

The goal in {\sc CAC-REC} is to compute a maximum weight subgraph $G'$ of $G$, $G'=((B,S), E' \cup C', W)$, that satisfies the following two constraints:
\begin{enumerate}
\item
For any $i$ in $B \cup S$, $d_{G'}(i) \leq D(i)$.
\item
For any $k$ in $S$, $|\{(i,j) | (i,k), (j,k) \in E', (i,j) \in C'\}| \leq t$.
\end{enumerate}

Here, $d_{G'}(i)$ denoted the degree of vertex $i$ in the subgraph $G'$.

We will next show that considering the conflict between buyers increases the complexity of the recommendation problem significantly. 

\subsection{NP-Hardness Result for CAC-REC}
\label{sec:np}

\nop{We investigate the complexity of P3. We prove the following result.}
In this section, we provide strong evidence that {\sc CAC-REC} is highly unlikely to have an efficient (i.e, polynomial time) algorithm by showing that it is NP-hard.

\begin{theorem}\label{thm:main}
{\sc CAC-REC} is NP-hard.
\end{theorem}

\begin{proof}
We give a polynomial-time reduction from the NP-hard problem {\sc Revenue Maximization in Interval Scheduling}.\\

\noindent{\sc Revenue Maximization in Interval Scheduling (RMIS)}

\medskip
{\bf\em Instance:} A set $M = \{m_1,m_2,\ldots,m_t\}$ of $t$ machines and a set $J=\{j_1,j_2,\ldots,j_n\}$ of $n$ jobs. For each job $j$ in $J$, we are
given three parameters: (1) $S(j)$, the set of machines on which $j$ can be processed (2) $R(j,-)$, the set of possible revenues obtained when job $j$ is processed on different machines (3) $I(j)$, the time interval during which job $j$ must be processed.\\

{\bf\em Goal:} Find a feasible schedule of a subset of jobs on the machines that maximizes the total revenue of the jobs scheduled.\\

\noindent We now describe a reduction from {\sc RMIS} to {\sc CAC-REC}. Given an instance I of the revenue maximization problem, construct a graph $G(I)$, which is an instance of {\sc CAC-REC}. Define $G(I) =\left<(J,M), E \cup C, W\right>$ with $E \subseteq J \times M$, $C \subseteq J \times J$ and weights $W: E \rightarrow \mathbb{R}^+$ as follows. 

\begin{itemize}
\item $E =\{(j_k,m_l) | m_l \in S(j_k)\}$. 
\item $C = \{(j_k,j_l) | I(j_k) \cap I(j_l) \neq \emptyset \}$.
\item $W(j_k,m_l) = R(j_k, m_l)$.
\item $D(j_k) = 1$ for all $k$ and $D(m_l) = n$ for all $l$.
\item $t$ = 0.
\end{itemize}

We now explain the reduction above. There is an edge between job $j_k$ and machine $m_l$ if $m_l$ belongs to $S(j_k)$, the set of machines in which job $j_k$ can be processed. There is a conflict edge between job $j_k$ and job $j_l$ if their time interval for processing overlap. The weights on an edge $(j_k,m_l)$ represent the revenue obtained if job $j_k$ is processed on machine $m_l$. Since each job can be assigned to at most one machine, their degree constraints are set to 1. There is no constraint on the number of jobs assigned to any machine and hence their degree constraint is set to $n$. Finally, there must be no conflict between two jobs assigned to same machine. Therefore, $t$ is set to 0.

It can be easily seen that an optimal solution for $G(I)$, an instance of {\sc CAC-REC} yields an optimal solution for I, an instance of RMIS. In other words, a maximum weight subgraph of $G(I)$ satisfying the degree constraints and conflict constraints as described above exactly corresponds to a revenue maximizing schedule in $I$. Furthermore, we observe that the reduction above is a polynomial-time reduction. Therefore we conclude that {\sc CAC-REC} is NP-hard.
\end{proof}

\subsection{A SDP Algorithm for CAC-REC} \label{sec:sdp_algorithm}
In the previous section, we showed that {\sc CAC-REC} is NP-hard. In this section and next, we will design two efficient algorithms for {\sc CAC-REC} that provide high-quality solutions that are close to optimal. 

Our first algorithm for {\sc CAC-REC} is based on a semidefinite programming based approach. To understand the motivation for this approach, recall that we described a LP formulation for {\sc C-REC} in Section~\ref{P1}. Using the terminology from Section~\ref{P1} and~\ref{sec:np}, the conflict constraint can be described as follows:

\begin{equation}\label{eqt:t}
\sum_{(j_k,j_l) \in C} x_{ki}x_{li} \leq t \mbox{  $\forall$ $i \in S$} %{\color{red}S}$}
\end{equation}
%$$.$$

\nop{???}
That is, the conflict constraint is quadratic. We use a single $t$ for illustration purpose. In practice, different sellers can have different values of $t$. We will now show how to formulate {\sc CAC-REC} as a semidefinite program. Define a $mn \times mn$ symmetric matrix $\mathbb{Y} = XX^T$ where $X$ is as in Section~\ref{P1}.  The CAC-REC problem can be described as

\begin{equation}\label{prob:CARec}
\begin{aligned}
&  \max
& Trace(\mathbb{W}\mathbb{Y}) \\
& \textit{s.t.}
& Trace(\mathbb{D}^b_i\mathbb{Y}) \leq D(i), & \forall~i \in B  \\
&& Trace(\mathbb{D}^s_i\mathbb{Y}) \leq D(i), & \forall~i \in S  \\
&& Trace(\mathbb{C}_i\mathbb{Y}) \leq t, & \forall~i \in S\\ 
&& \mathbb{Y} = XX^T \succeq 0 & %\mathbb{Y} = XX^T
\end{aligned}
\end{equation}
where $\mathbb{W}, \mathbb{D}^b_i, \mathbb{D}^s_i$ and $\mathbb{C}_i$ are suitably defined $mn \times mn$ symmetric matrices described below. 

\begin{enumerate}
\item 
$\mathbb{W}$ is a diagonal matrix with diagonal weights $w_{ij}$. 
\item
$\mathbb{D}^b_i$ is diagonal matrix with a 1 for row indexed by $(i,j)$ if $(i,j) \in E$ and 0 otherwise.
\item
$\mathbb{D}^s_i$ is diagonal matrix with a 1 for row indexed by $(k,i)$ if $(k,i) \in E$ and 0 otherwise.
\item
Finally, $\mathbb{C}_i$ is a matrix with entries 1/2 and 0. An entry indexed by row $(j,i)$ and column $(k,i)$ is equal to 1/2 if $(j,i), (k,i) \in E$ and $(j,k) \in C$. It is 0 otherwise. 
\end{enumerate}

\nop{Let us start by defining an indicator variable $x_{ij}$ for any edge $(i,j) \in E$.

\begin{eqnarray*}
x_{ij} & = & 1 \mbox{ if $(i,j) \in E'$} \\
       & = & 0 \mbox{ otherwise}
\end{eqnarray*}

Then the objective function to be maximized in P3 is given by $$\sum_{(i,j) \in E}w_{ij}x_{ij}$$ 

Let us now describe the constraints. The degree constraints are described by 

$$ \forall v \in B \cup S, \sum_{(v,v') \in E}x_{v,v'} \leq D(v).$$

The conflict constraints are described by

$$ \forall v \in S, \sum_{(v,v') \in C}x_{v,s}x_{v',s} \leq t.$$

Let $\hat{x}=[x_{ij}]$ be a column vector and $X=\hat{x}\hat{x}^T$. Then, by definition, $X$ is positive semidefinite. Using $X$ and the fact that 
$x_{ij}^2 = x_{ij}$, we can rewrite the formulation as 

\medskip
\noindent
Maximize $WX $\\
subject to \\
$D_v \cdot X \leq D(v)$ \mbox{ $\forall v \in B \cup S$} \\
$C_v \cdot X \leq t$ \mbox{ $\forall v \in S$}\\ 
$X \succeq 0$

where $W, D_v$ and $C_v$ are suitably defined symmetric matrices. 
}
\noindent
Our SDP based algorithm for CAC-REC is as follows:
\begin{enumerate}
\item
Solve the semidefinite program relaxation to obtain optimal solution $\mathbb{Y}$.
\item
Using the Cholesky decomposition~\cite{Press} of $\mathbb{Y}$, obtain the vectors $x_{ij}$ corresponding to $\mathbb{Y}$.
\item
Use a two-step rounding procedure, random projection followed by threshold rounding, to obtain $0,1$ values for $x_{ij}$.
\end{enumerate}

In step (1) of our SDP algorithm, the SDP described above is solved using a generic SDP solver. The output of step (1) is a semidefinite matrix $\mathbb{Y}$. In step (2) of our algorithm, we use a well-known fact that any semidefinite matrix $\mathbb{Y}$ can be written as $\mathbb{Y}=\mathbb{V}\mathbb{V}^T$ where $\mathbb{V}$ is a $mn \times mn$ lower triangular matrix. This decompostion is known as Cholesky decomposition of $\mathbb{Y}$~\cite{Press}. The columns of $\mathbb{V}$ give us a vector solution for the variables $x_{ij}$. Thus, the output of step (2) of our algorithm are $mn$ vectors, one for each $x_{ij}$. These vectors correspond to the optimal solution of the SDP. In the last step of our algorithm, we convert the vectors $x_{ij}$ to integral $\{0,1\}$ values using a two-step rounding procedure. In the first step, we convert $x_{ij}$'s to fractional values by a random projection. That is, we pick a random vector $x$ of dimension $mn$ by picking each of its coordinates from the normal distribution $N(0,1)$ and define each $x_{i,j}$ as the length of its projection on to $x$. Finally, we sort $x_{ij}$ and round each non-zero value to 1 provided doing so does not violate the degree constraints or the conflict constraints. Otherwise, we set it to 0.

\nop{We study the experimental performance of this algorithm on real world data sets in the experiments section.}

\subsection{ILP Formulation of CAC-REC}
Solving the SDP formulation requires large enough physical memory to store the $mn \times mn$ matrix $\mathbb{Y}$. In practice (e.g., the eBay purchase graph of a certain category), large values for $m$ (the number of buyers) and $n$ (the number of sellers) inevitably restrict the applicability of the SDP approach. This limitation, however, can be alleviated if we could model CAC-REC as an integer linear programming (ILP) problem. 

In order to achieve this goal, we introduce a new $0$-$1$ variable $z_{i, (j,k)}$ to formulate Inequality~\ref{eqt:t} as a linear constraint. For each seller $i$, $z_{i, (k,l)}$ equals $1$ if and only if there is a conflict edge between two buyers $k$ and $l$, and both edge $e_{ki}$ and $e_{li}$ are recommended in the graph. Using the terminology from Section~\ref{P1} and~\ref{sec:np}, this constraint can be described as follows:

%\begin{equation}
\begin{align}
(1 - x_{ki}) + (1 - x_{li}) + z_{i, (k,l)} &\geq 1 \mbox{  $\forall$ $i \in S, \quad \forall k,l \in B$} \\
x_{ki} + x_{li} - 2 z_{i, (k,l)} &\geq 0 \mbox{  $\forall$ $i \in S, \quad \forall k,l \in B$} \\
\sum_{(k,l) \in C_i} z_{i, (k,l)} &\leq t \mbox{  $\forall$ $i \in S$} 
\end{align}

In Constraint (6), $C_i$ is defined as follows:
$C_i = \{ (k,l) \in C | (k,i) \in E \wedge (l,i) \in E \}$.
%\end{equation}

%as Problem~\ref{prob:CoRec}
The linear conflict constraints can be easily incorporated into C-REC (Problem~\ref{prob:CoRec}). Let $c$ denote the total number of conflict constraints with respect to all sellers in the graph. Then we can obtain a linear programming formulation, where $X = [x_{ij}]^T$ is a ($mn+c$)-dimensional column vector of $0$-$1$ variables. In addition, matrix $\mathbb{A}$ and vectors $W$ and $D$ can be changed accordingly. 

By eliminated the need to store large dimensional matrices, now we can use a ILP solver to tackle CAC-REC problems of larger sizes. As opposed to C-REC, obtaining an integer solution in CAC-REC is NP-hard. In order to further improve efficiency, we use a rounding procedure after solving the linear program relaxation. Our LP based algorithm for CAC-REC is as follows:
\begin{enumerate}
\item
Solve the linear program relaxation to obtain optimal solution $X$.
\item
Sort the first $mn$ elements of $X$ from largest to smallest. We round each non-zero value to 1 provided doing so does not violate the degree constraints or the conflict constraints. Otherwise, we set it to 0.
\end{enumerate}

\subsection{A Greedy Algorithm for CAC-REC} 

In this section, we describe and study the performance of a simple greedy algorithm for this problem. This algorithm has the advantage that it is highly scalable and provides good quality solutions in practice. 

The greedy algorithm, denoted as GREEDY, for CAC-REC is as follows:
\begin{enumerate}
\item
Sort all the edges in $E$ by weight from largest to smallest.
\item
To construct the maximum weight subgraph $G'$, consider every edge in the sorted list. Add this edge to $G'$ if doing so
does not violate any degree constraint or conflict constraint.
\item
Continue until we reach the end of the sorted list.
\end{enumerate}

We will now prove a theoretical guarantee on the performance of GREEDY.

\begin{theorem}
Let $d= \max_{v \in B} |\{(v,v') | (v,v') \in C\}|$. Algorithm GREEDY is a $(2+d)$- approximation algorithm.
\end{theorem}
\begin{proof}
We use the concept of a {\em $k$-extendible system} to provide performance guarantees of a greedy algorithm. Mestre introduced the notion of a $k$-extendible system in his study of the performance of the greedy technique as an approximation algorithm~\cite{M06}. 

\begin{definition}[$k$-Extendible System {\cite{M06}}]
Let $U$ be a finite set and ${\cal F}$, ${\cal F} \subseteq 2^U$, be a collection of subsets of $U$. Set system $(U,\cal F)$ is called a {\em $k$-extendible system} if it satisfies the following properties:
\begin{enumerate}
\item
{\em Downward-closure}: If $A \subseteq B$ and $B \in \cal F$, then $A \in \cal F$.
\item
{\em Exchange}: Let $A,B \in \cal F$ with $A \subseteq B$, and let $x \in U-B$ be such that $A \cup \{x\} \in \cal F$. Then there exists $Y \subseteq B-A$, $|Y| \leq k$, such that $(B-Y) \cup \{x\} \in \cal F$. In other words, let us start with any choice of two sets $A$ and $B$ such that $B$ is an extension of $A$. Suppose that there is an element $x$ such that 
the set $A$ with $x$ added to it also belongs to $\cal F$. Then we will be able to find a subset $Y$ inside $B$ of size at most $k$ such that if we remove the elements of $Y$ from $B$ and add the element $x$ to the resulting set, it will also belong to the collection $\cal F$. 
\end{enumerate}
\end{definition}

Informally, Mestre showed that if the set of all feasible solutions forms a $k$-extendible system, algorithm GREEDY gives a $k$-approximation algorithm. That is, on any instance, the solution output by GREEDY differs from the optimal solution by a multiplicative factor of at most $k$. We now state his result more formally. 
%is at least the value of the optimal solution

%is within a multiplicative factor of $k$ from the optimal.

\begin{figure*}[th]
\centering
\subfigure[Degree Constraint Ratio: $10\%$]{
\includegraphics[scale=0.15]{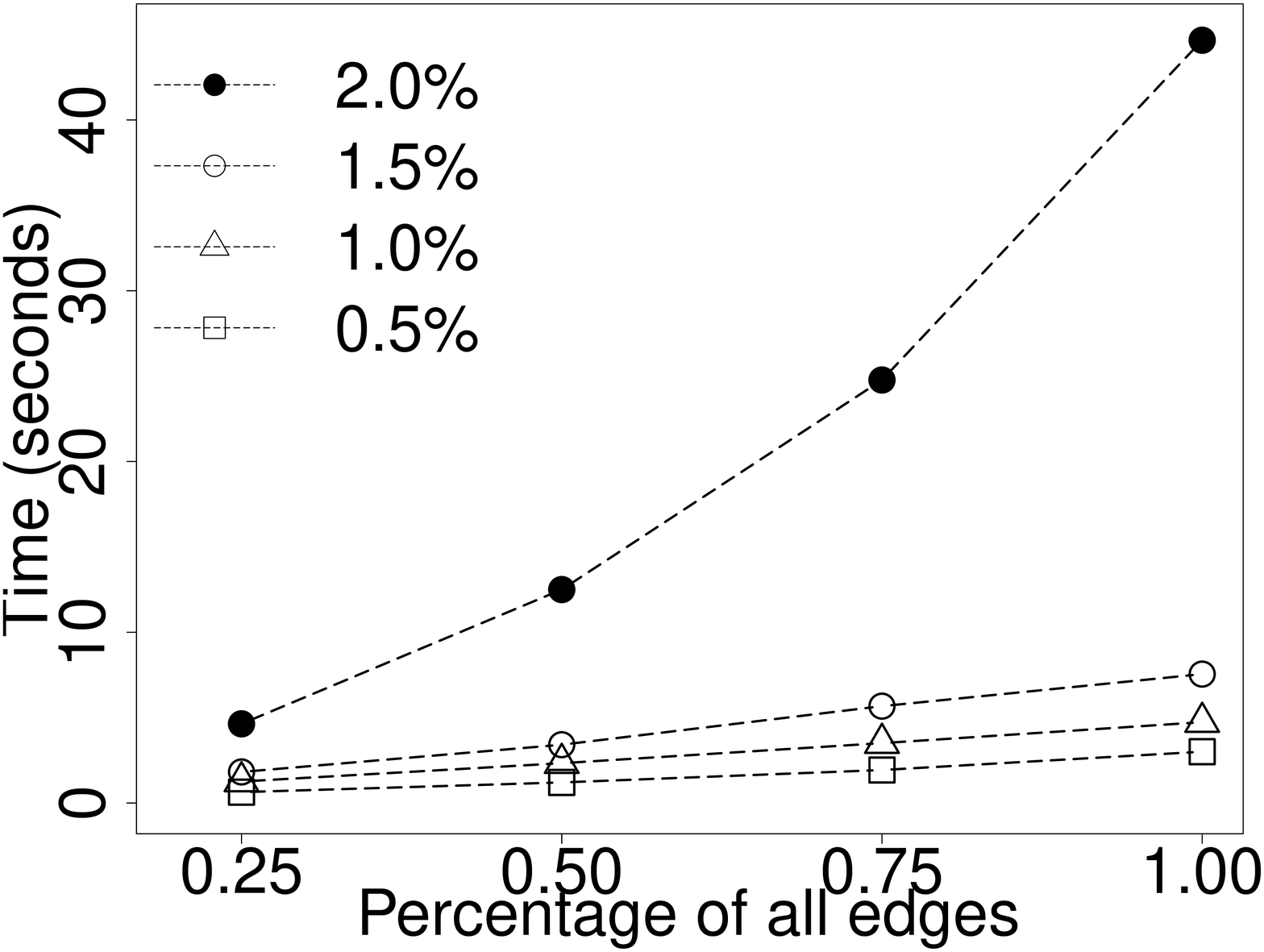}
}
\subfigure[Degree Constraint Ratio: $20\%$]{
\includegraphics[scale=0.15]{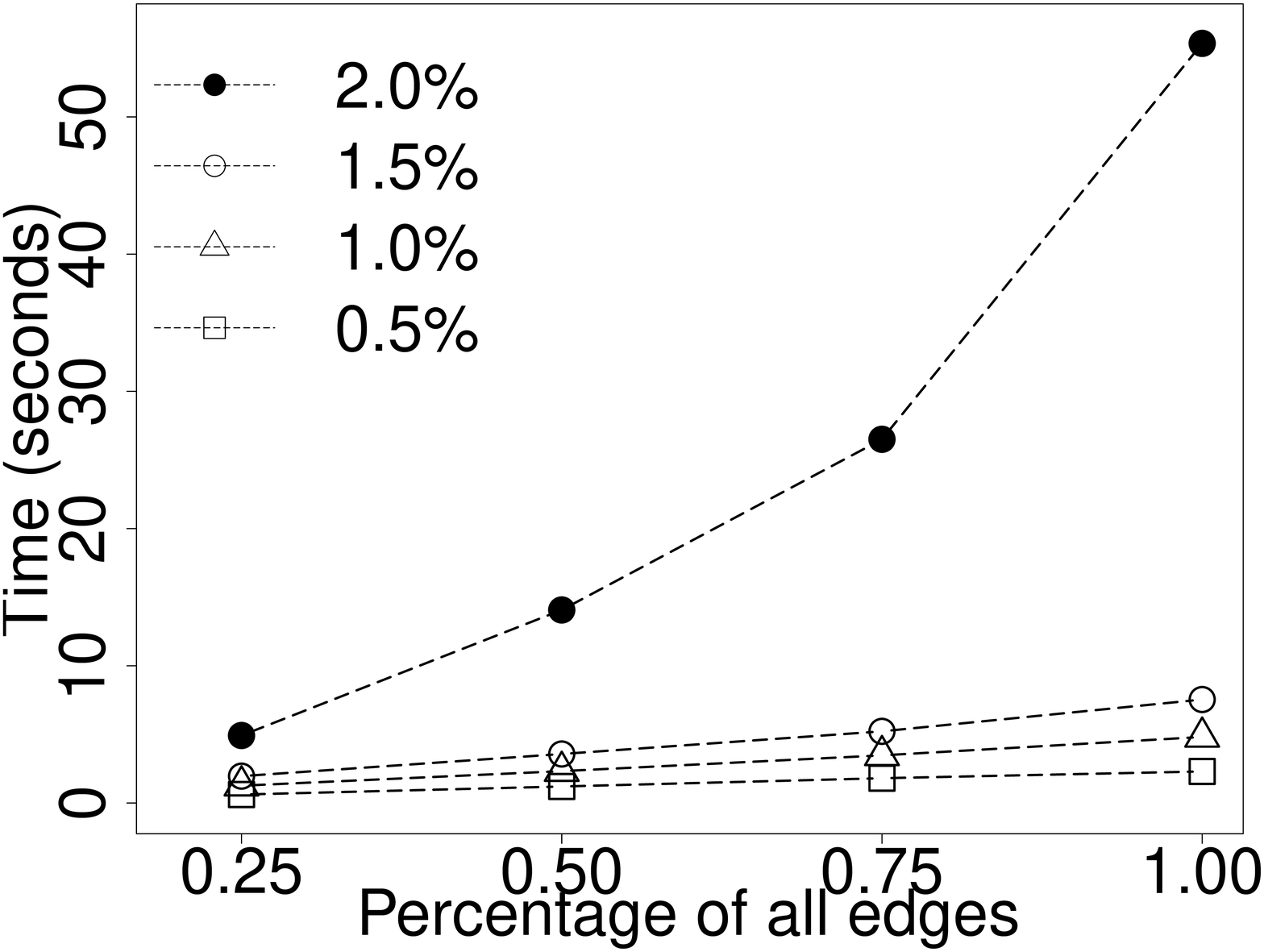}
}
\subfigure[Degree Constraint Ratio: $30\%$]{
\includegraphics[scale=0.15]{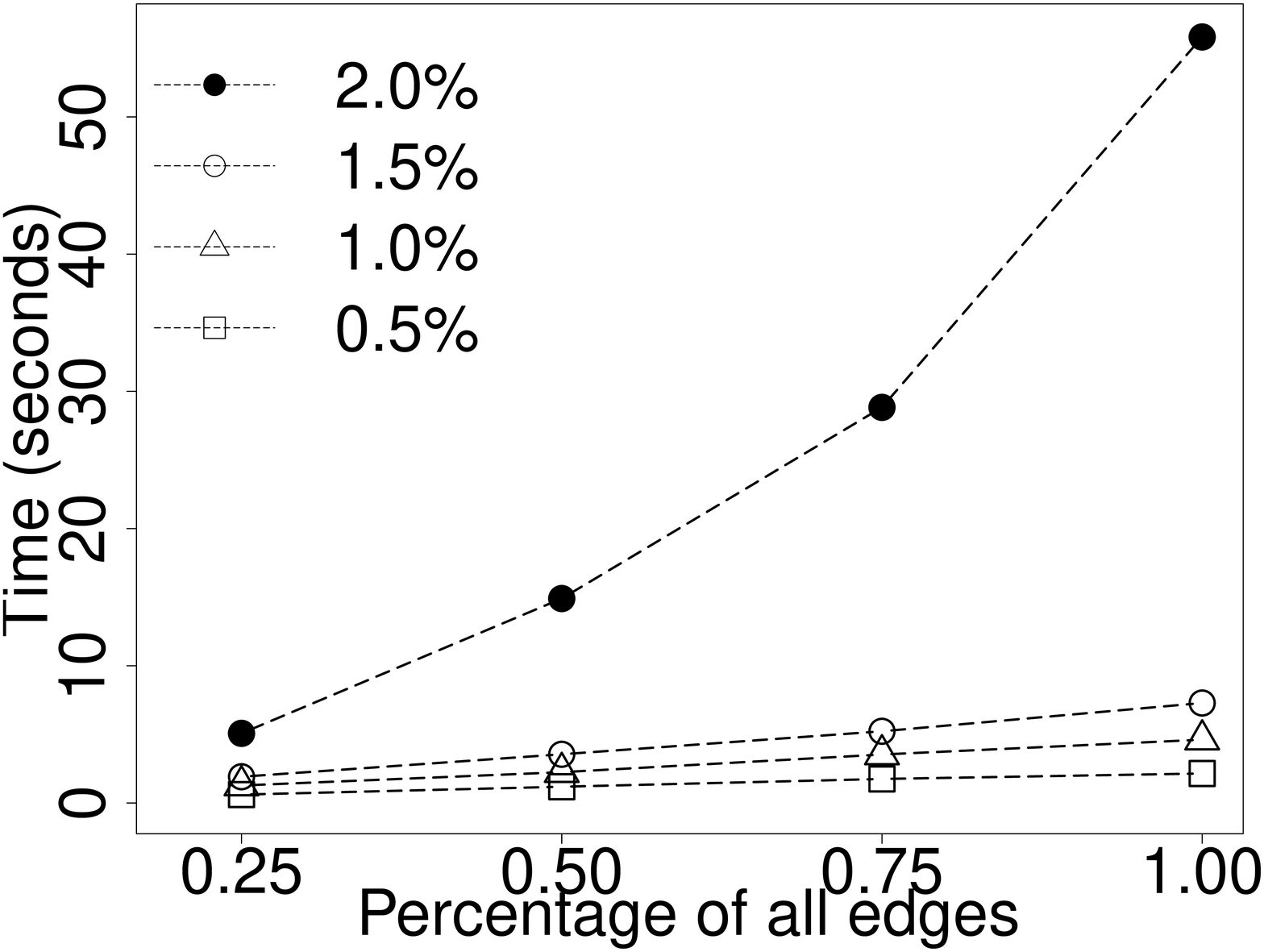}
}
\subfigure[Degree Constraint Ratio: $40\%$]{
\includegraphics[scale=0.15]{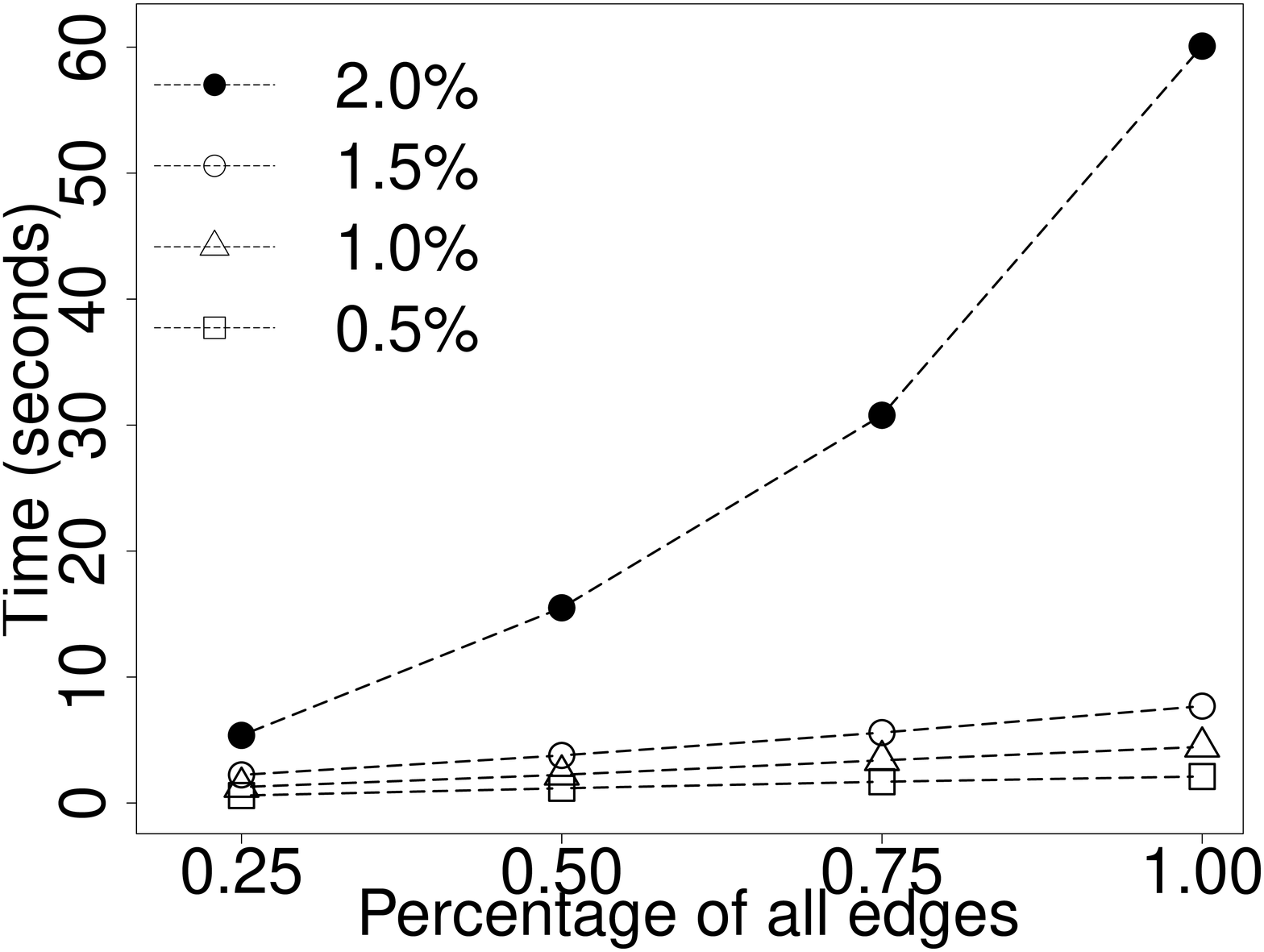}
}
\subfigure[Degree Constraint Ratio: $50\%$]{
\includegraphics[scale=0.15]{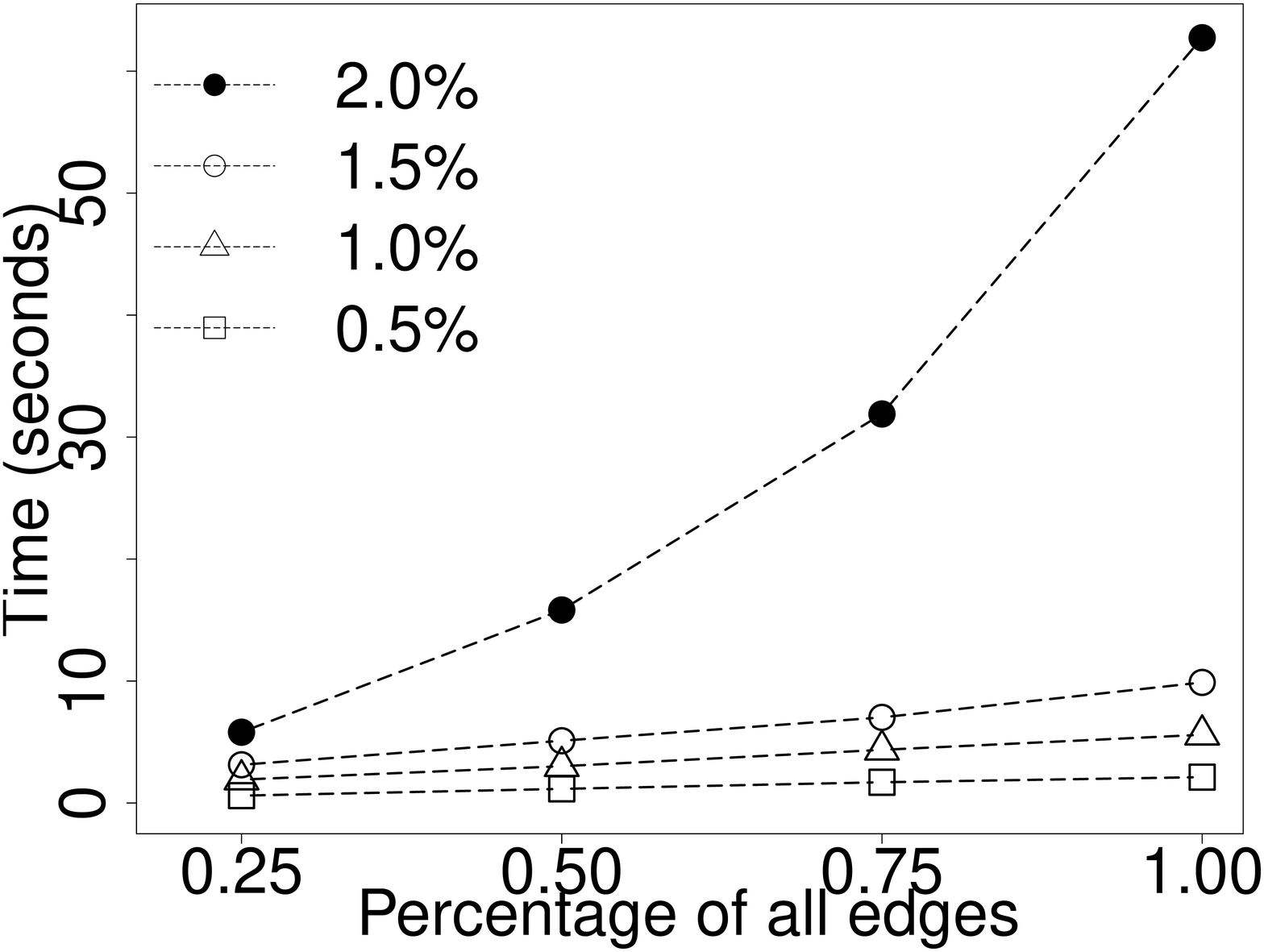}
}
\caption{C-REC running time of different degree constraint ratios. Similar trends are seen for running time in all scenarios, i.e., C-REC with a higher density requires a longer time to be solved and the running time increases when the size of the graph is scaled up.}
\label{fig:p1_time}
\end{figure*}

\begin{theorem}[Mestre~\cite{M06}]\label{thm:M06}
Let $(U, \cal F$) be a $k$-extendible system for some $k$. Let $W: U \rightarrow I\!\!R^{+}$ be a positive weight function on $U$. Then, the greedy algorithm gives a \nop{$\frac{1}{k}$-}$k$-approximation algorithm for the optimization problem that asks to determine $\max\limits_{F \in \cal F} W(F)$ where $W(F) = \sum\limits_{s \in F} W(s)$ for any $F \in \cal F$.
\end{theorem}

To apply this result to our problem, we will check that the set of all feasible solutions to CAC-REC forms a $(2+d)$-extendible system. For the CAC-REC problem,  $U = E$ and ${\cal F}$ be the set of all subgraphs of $G$ satisfying the degree and conflict constraints. Then it is easy to see that $(U,\cal F)$ is downward closed. That is, removing an edge from a feasible solution $H$ will always result in be a feasible solution as this will not cause any violation of constraints.

For the exchange property, consider the case when a new edge $e=(u,v)$, $u \in B$ and $v \in S$, is added to a feasible solution $H$.  We make two observations: (1) Adding $e$ could result in violation of degree constraint at $u$ and $v$. However, this can be rectified by removing two other edges, one incident on $u$ and other incident on $v$; (2) Adding $e$ could result in a violation of the conflict constraint at $v$. Rectifying this could require removing at most $d$ edges where  $d= \max_{v \in B} |\{(v,v') | (v,v') \in C\}|$. Therefore,
we obtain a $(2+d)$-extendible system.
\end{proof}
We remark that this analysis is worst-case. In practice, GREEDY shows far superior performance, as demonstrated in our later test with real-world dataset.

\begin{figure*}[htb]
\centering
\subfigure[Conflict Pair Ratio: $5\%$]{
\includegraphics[scale=0.17]{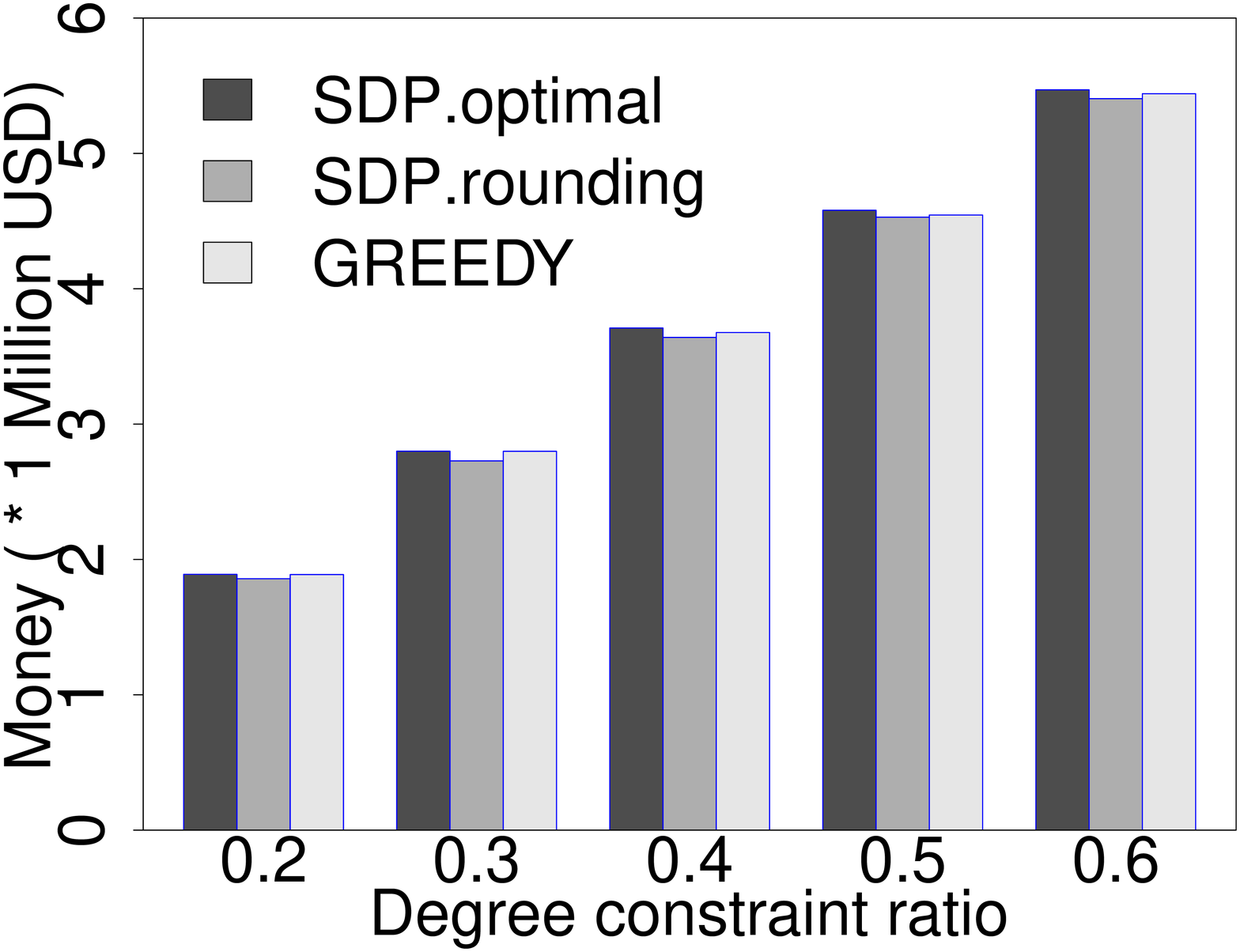}
}
\subfigure[Conflict Pair Ratio: $10\%$]{
\includegraphics[scale=0.17]{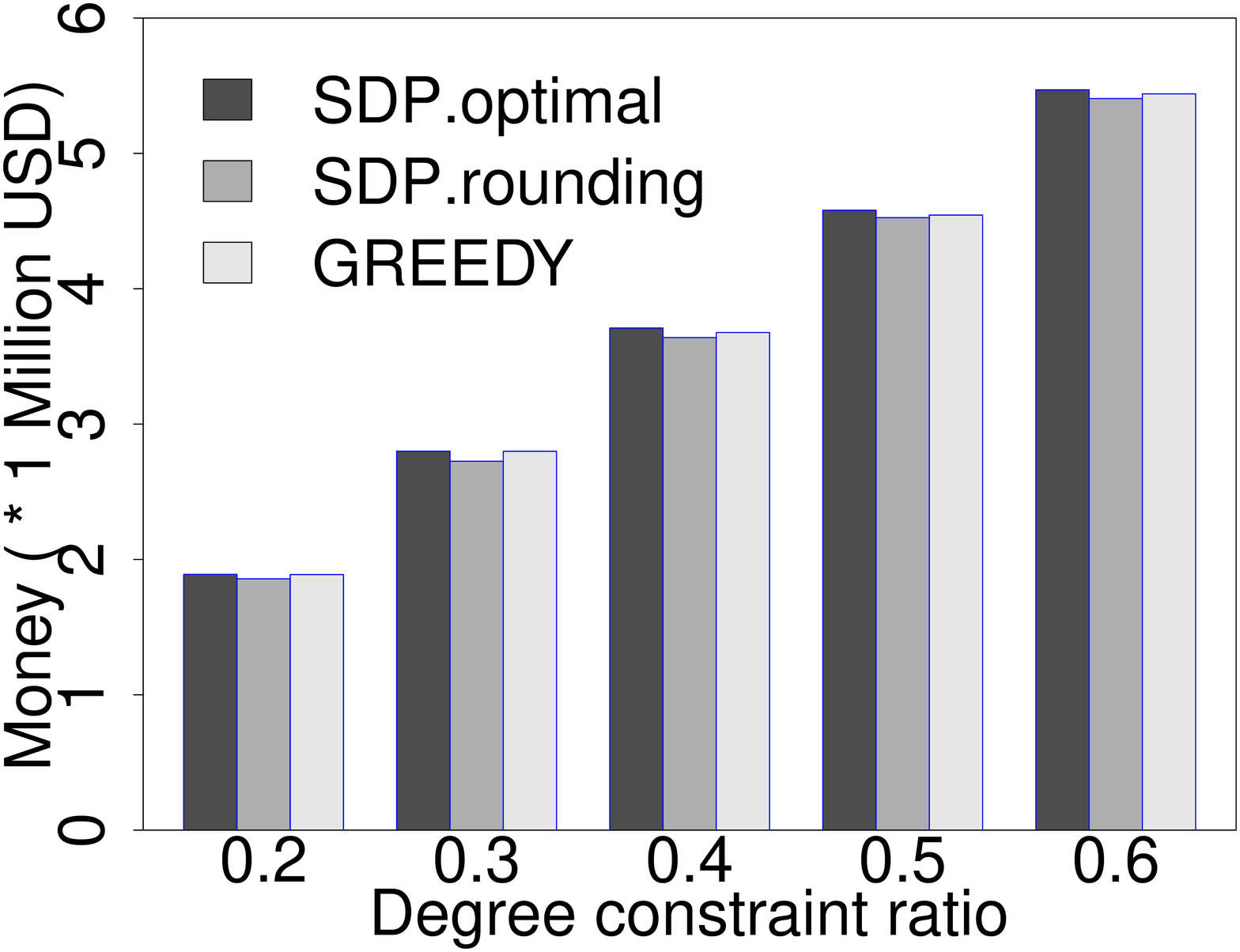}
}
\subfigure[Conflict Pair Ratio: $15\%$]{
\includegraphics[scale=0.17]{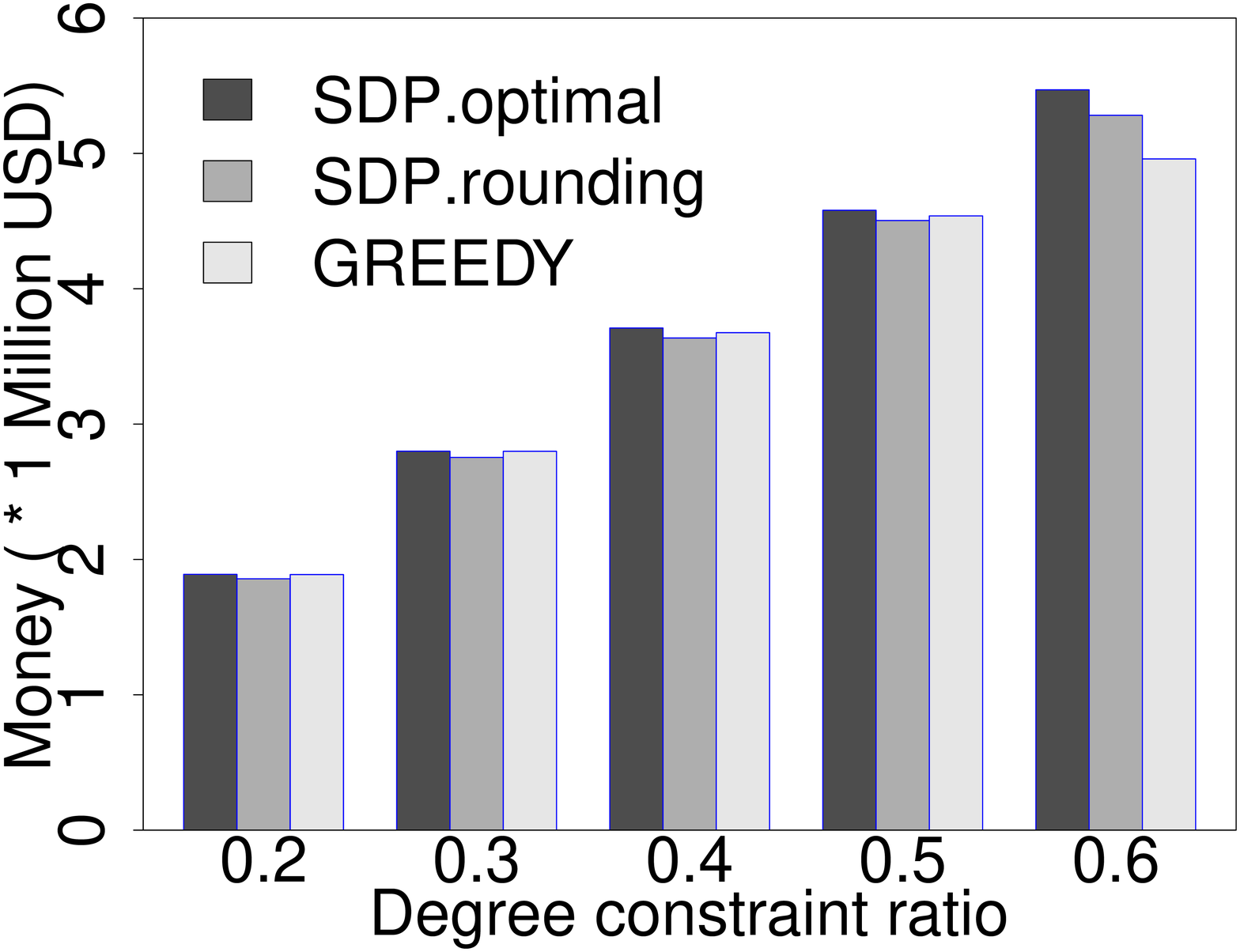}
}
\subfigure[Conflict Pair Ratio: $20\%$]{
\includegraphics[scale=0.17]{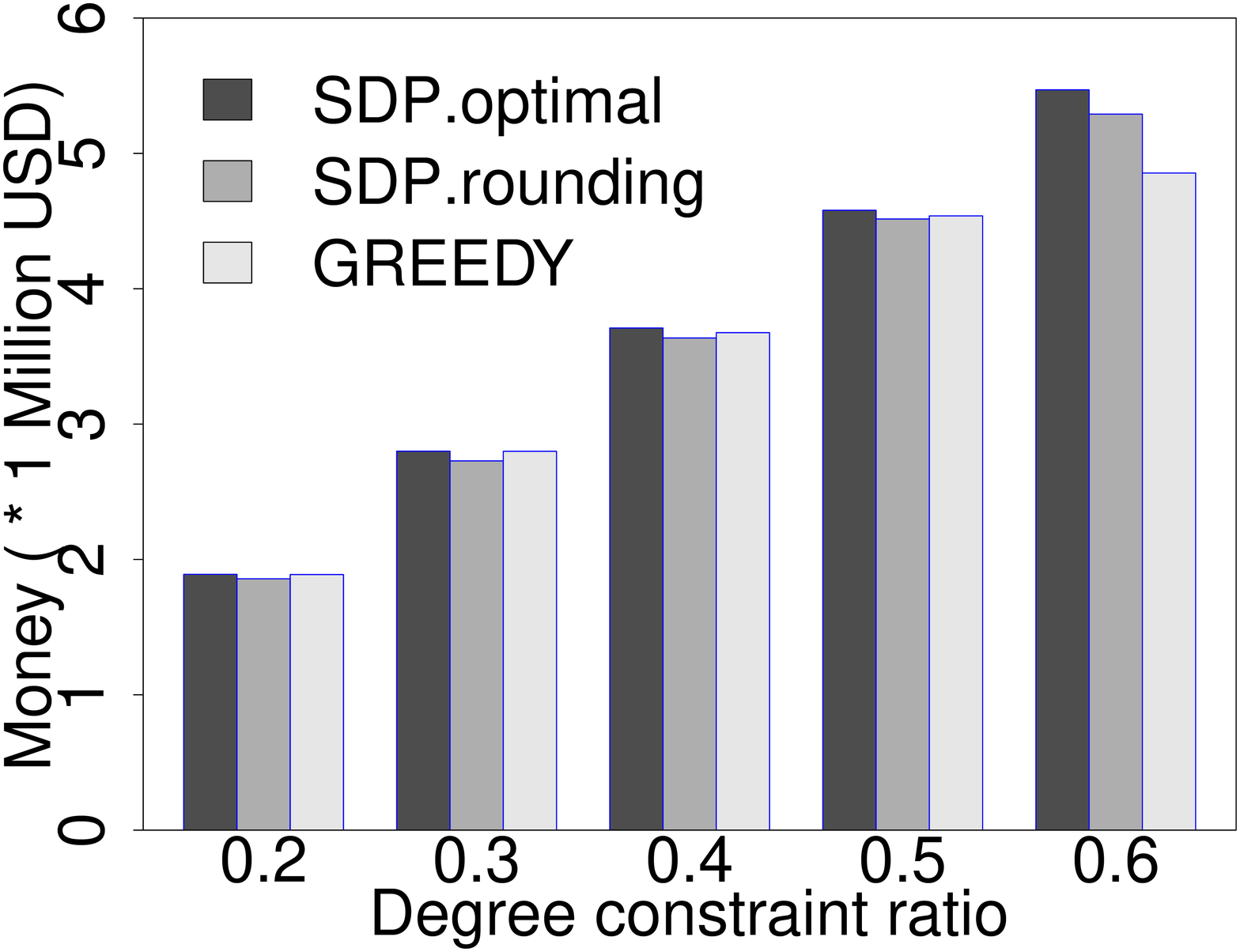}
}
%\subfigure[Conflict Pair Ratio: $60\%$]{
%\includegraphics[scale=0.116]{{P3_C0.6_money_solution}.eps}
%}
\caption{CAC-REC money solution of different conflict pair ratios. When degree constraint ratio and conflict pair ratio are both low, GREEDY shows close to optimal solutions while SDP with rounding is weaker.}
\label{fig:p3_money_solution}
\end{figure*}

\section{Experimental Evaluation}
\label{experiments}
In order to illustrate the proposed algorithms' optimality and scalability, we conducted comprehensive experiments on real-world datasets for {\em C-REC} and {\em CAC-REC}. The datasets (provided by Terapeak\footnote{Terapeak is an E-commerce company, helping sellers on eBay or Amazon to measure and boost their sales performance. http://www.terapeak.ca/}) include three-month transaction records across all categories of eBay Canada in 2013. We note that even though we used data from Terapeak, similar datasets can be easily collected via the eBay API. We use the sales data of a specific category (Cell phones and Accessories) in our study, which contains $18,742$ buyers and $1,884$ sellers after filtering out buyers who only purchased from a single seller and sellers who sold items to no more than $5$ distinct buyers. Then we sort buyers, as well as sellers, by the total monetary purchases and total profits, respectively. The weight on each edge comes from the real-world sales data and will be explained in more detail in the following experiments.

\nop{Very Questionable Statements: Due to the fact that the existing connections (i.e., transactions) between buyers and sellers are sparse, the number of edges is not enough to test our algorithms. Hence, we create edges by ourselves for the illustration purpose. The weights on edges, however, come from the real-world sales data.}

All experiments were run on a 64-bit Ubuntu 12.04 desktop of $3.40$GHz * 8 Intel Core i7 CPU and $3.8$ GB memory.

\subsection{C-REC}
\label{sec:c-rec_exp}
%{\color{red} Since it is guaranteed that {\em C-REC} has integral solutions, we use a fast linear programming (LP) solver,}
We use a fast linear programming (LP) solver, Gurobi\footnote{http://www.gurobi.com/} for Matlab, to solve {\em C-REC} at different scales. The purpose is to observe the scalability of the LP approach.

Adding edges between sorted buyers and sellers, we create four bipartite graphs of different density settings, roughly at $0.5\%$, $1.0\%$, $1.5\%$ and $2.0\%$. Edges are generated in a manner that each seller has the same number of connected buyers. For example, if the density is $0.5\%$, the first seller (top ranked) is connected with the first $90$ buyers (from $1$ to $90$), and the second seller is connected with buyers from 11 to 100, and so on. We also extract three subsets of different sizes for each bipartite graph, i.e., using $25\%$, $50\%$ and $75\%$ of the total number of edges (the number of buyer and seller nodes decreases accordingly). The weight of an edge is the sum of the buyer's total monetary purchases (on all sellers within the category) and the seller's total profits (from all buyers within the category). For each buyer (seller) in the bipartite graph, the degree constraint ratio is the proportion of the maximum number of recommended sellers (buyers) among all candidates. We choose constraint ratios from the set $\{0.1, 0.2, 0.3, 0.4, 0.5\}$ to illustrate their impact on the running time. Every node (buyer or seller) in the experiment has the same degree constraint ratio.

Figure~\ref{fig:p1_time} shows the running time of different experimental settings. The four curves in each subfigure correspond to results of bipartite graphs with different densities, i.e., $0.5\%$, $1.0\%$, $1.5\%$, and $2.0\%$, respectively. The final result is computed as the average of five runs. We make the following observations:
\begin{enumerate}
\item The curve of a higher density is always positioned above the one of a lower density, indicating that it takes a longer time to solve a {\em C-REC} of a higher density. Specifically, the $2.0\%$ density {\em C-REC} requires much more time compared to others of smaller density.
\item The running time increases when we scale up the size of the bipartite graph. For example, the rising curve of the $2.0\%$ density becomes steeper as we add more edges. We observed a similar trend in other density curves, but they are not as obvious as the $2.0\%$ one.
\item The running time slightly changes when the degree constraint ratio becomes larger. For example, the running times of the $2.0\%$ curve of all edges for each degree constraint ratio (subfigures (a), (b), (c), (d), and (e)) are $44.67$, $55.36$, $55.83$, $60.08$ and $62.74$ seconds, respectively. 
\end{enumerate}

From Figure~\ref{fig:p1_time}, we observe that no matter how we change the size of the bipartite graph and degree constraints, it only takes about one minute to solve the largest instance of {\em C-REC}. \nop{Considering that these density values are significantly larger than the one of the real-world data, which is approximately $0.02\%$,}The results suggest that the LP approach for {\em C-REC} is highly scalable. 

\nop{
\subsection{Problem 2}
\begin{figure}[h]
\centering
\subfigure[Complete Graph of 1000 Buyers]{
\includegraphics[scale=0.15]{P2_1000_money_time.eps}
}
\subfigure[Complete Graph of 1500 Buyers]{
\includegraphics[scale=0.15]{P2_1500_money_time.eps}
}
\caption{Problem 2: Money Time}
\label{fig:p2_money_time}
\end{figure}

\begin{figure}[h]
\centering
\subfigure[Complete Graph of 1000 Buyers]{
\includegraphics[scale=0.15]{P2_1000_money_solution.eps}
}
\subfigure[Complete Graph of 1500 Buyers]{
\includegraphics[scale=0.15]{P2_1500_money_solution.eps}
}
\caption{Problem 2: Money Solution}
\label{fig:p2_money_solution}
\end{figure}

\begin{figure}[h]
\centering
\subfigure[Complete Graph of 1000 Buyers]{
\includegraphics[scale=0.15]{P2_1000_rank_time.eps}
}
\subfigure[Complete Graph of 1500 Buyers]{
\includegraphics[scale=0.15]{P2_1500_rank_time.eps}
}
\caption{Problem 2: Money Time}
\label{fig:p2_rank_time}
\end{figure}

\begin{figure}[h]
\centering
\subfigure[Complete Graph of 1000 Buyers]{
\includegraphics[scale=0.15]{P2_1000_rank_solution.eps}
}
\subfigure[Complete Graph of 1500 Buyers]{
\includegraphics[scale=0.15]{P2_1500_rank_solution.eps}
}
\caption{Problem 2: Money Solution}
\label{fig:p2_rank_solution}
\end{figure}
}

\subsection{CAC-REC}
Now we introduce conflict constraints between pairs of buyers. 
\subsubsection{SDP Formulation of CAC-REC}
As discussed in Section~\ref{P3}, the SDP approach of {\em CAC-REC} can be approximately solved by a SDP solver plus a rounding procedure or by a greedy algorithm. We use SDPT3~\cite{Toh99sdpt3--, DBLP:journals/mp/TutuncuTT03} as the SDP solver and we write a Python script to implement the GREEDY algorithm. We compare the solution of each method, optimal SDP (solution without rounding), rounded SDP, and GREEDY.

\begin{figure*}[htb]
\centering
\subfigure[Conflict Pair Ratio: $5\%$]{
\includegraphics[scale=0.17]{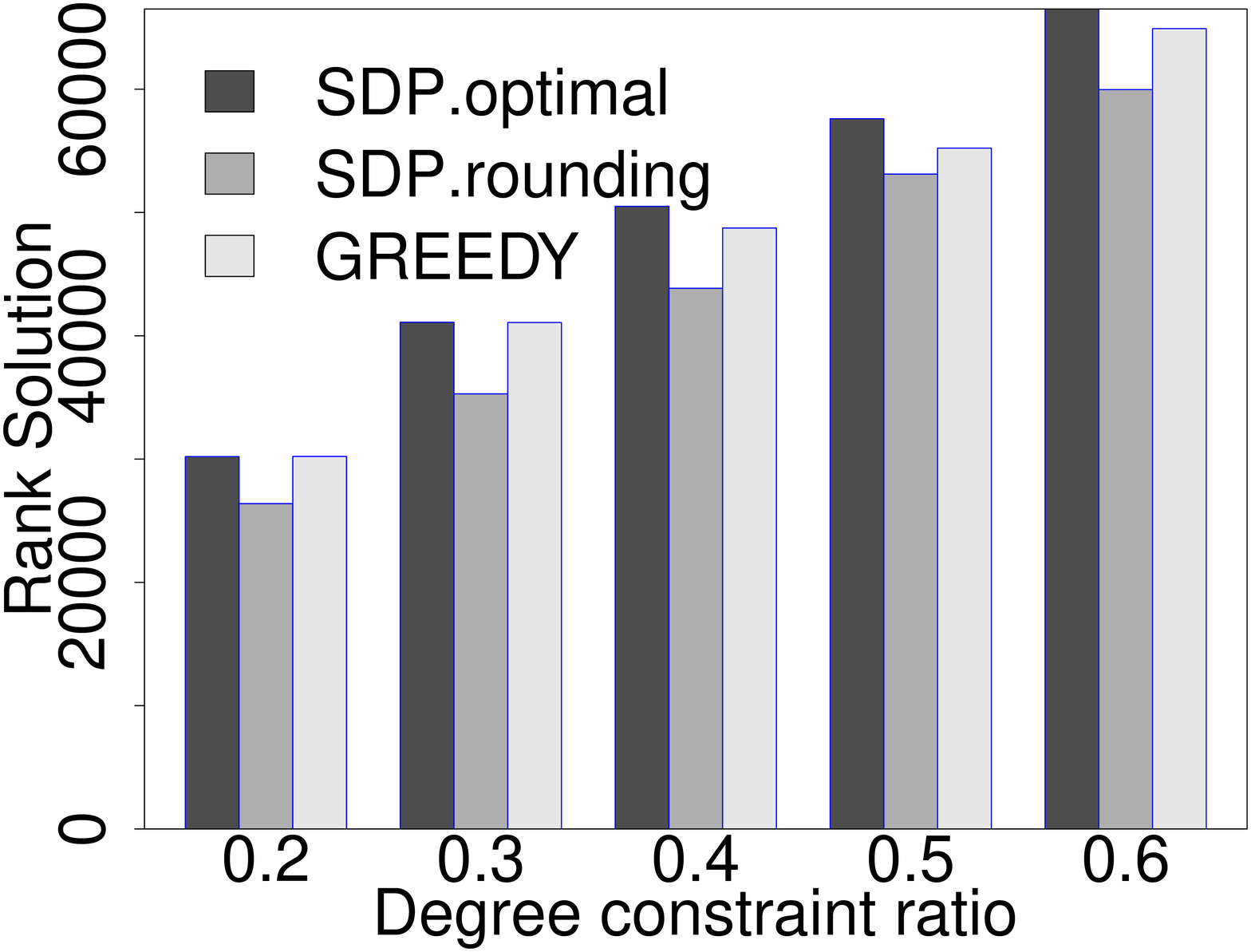}
}
\subfigure[Conflict Pair Ratio: $10\%$]{
\includegraphics[scale=0.17]{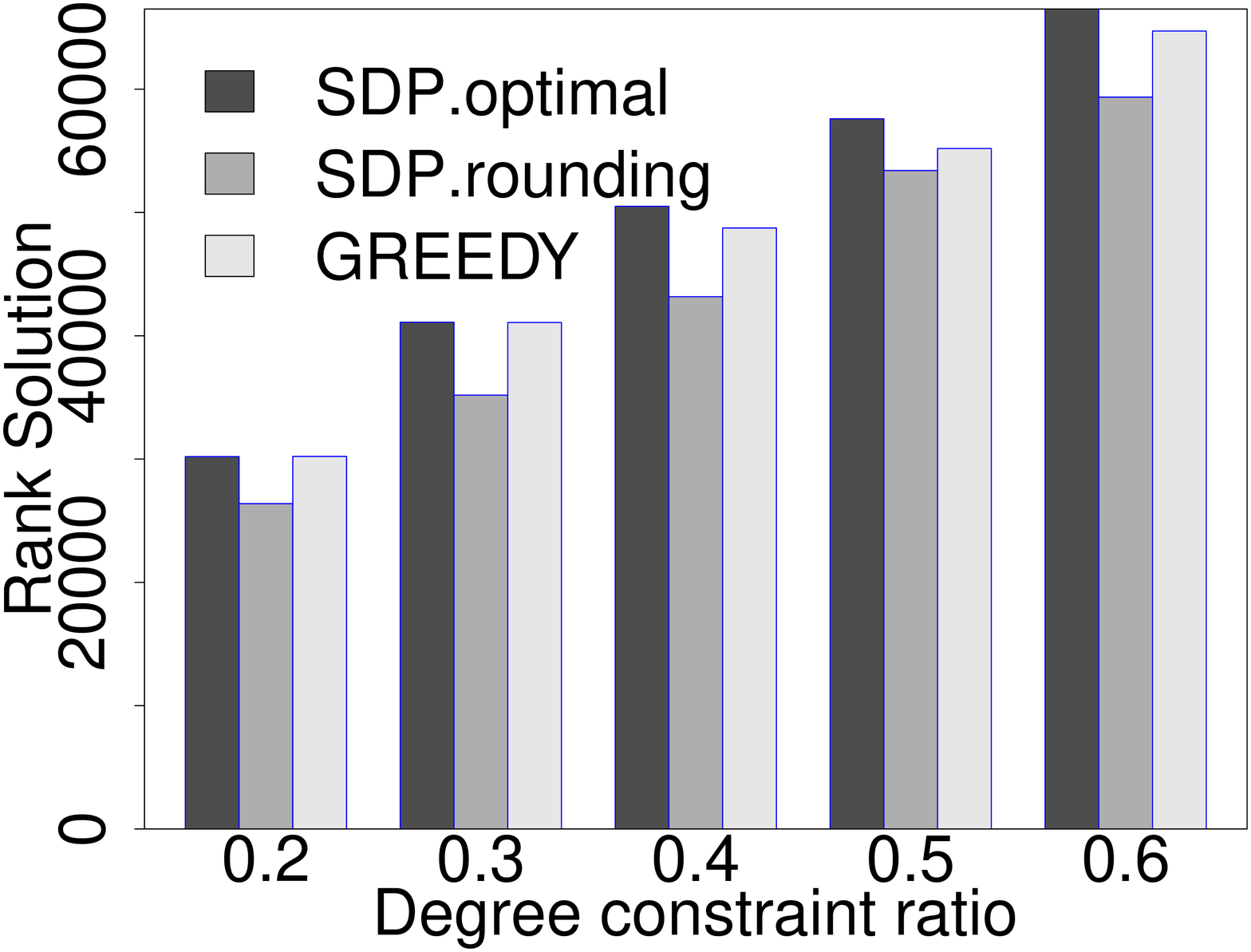}
}
\subfigure[Conflict Pair Ratio: $15\%$]{
\includegraphics[scale=0.17]{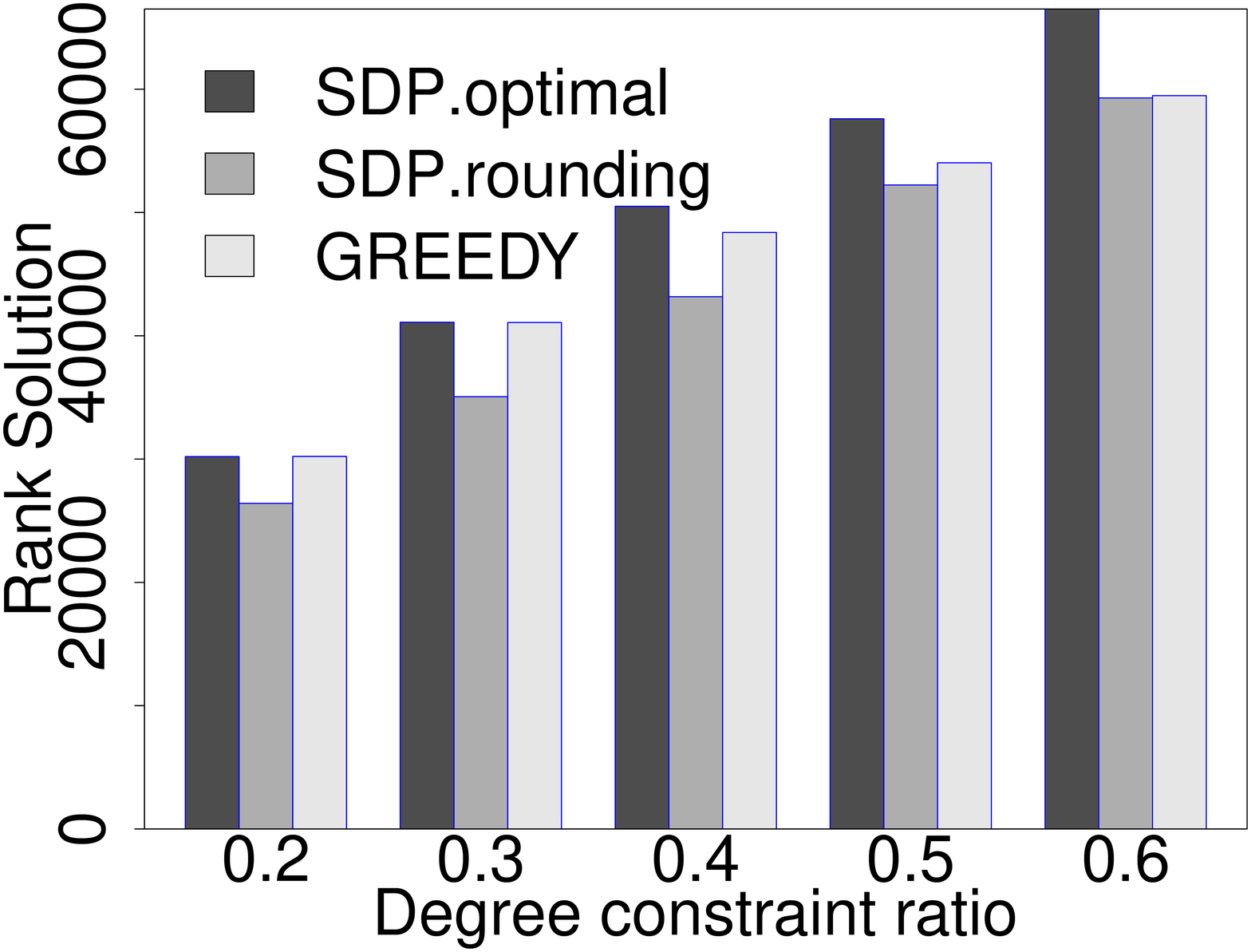}
}
\subfigure[Conflict Pair Ratio: $20\%$]{
\includegraphics[scale=0.17]{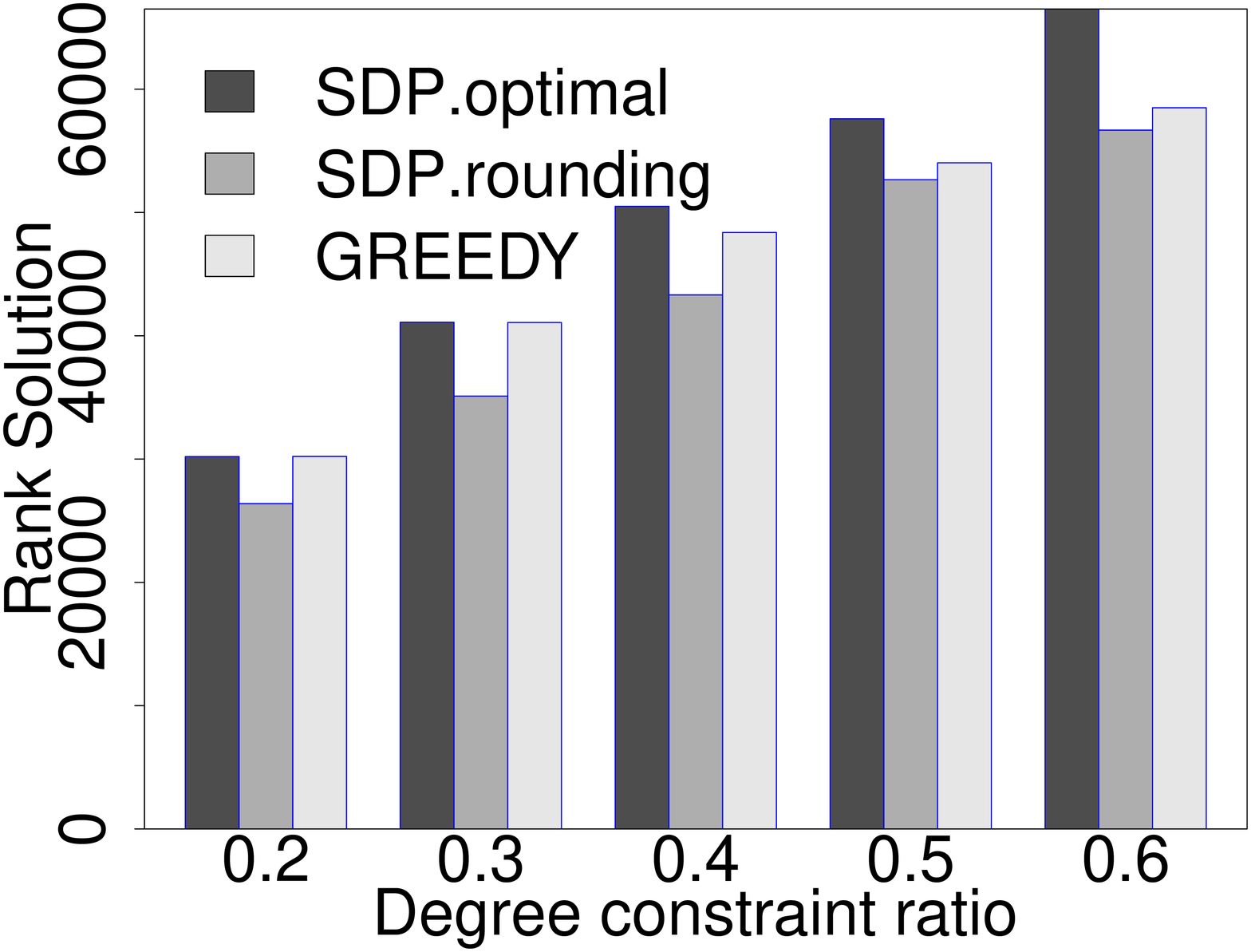}
}
%\subfigure[Conflict Pair Ratio: $60\%$]{
%\includegraphics[scale=0.116]{{P3_C0.6_rank_solution}.eps}
%}
\caption{CAC-REC rank solution of different conflict pair ratios. GREEDY achieves close to optimal performance in all cases. The solutions obtained by SDP with rounding are slightly worse compared to GREEDY.}
\label{fig:p3_rank_solution}
\end{figure*}

The applicability of the SDP approach for {\em CAC-REC} severely suffers from the limitation of physical memory, i.e., the need to store a large-dimensional matrix~\cite{DBLP:journals/mp/KocvaraS07}. %In the real-world scenario, the dimension of the matrix can easily reach a few thousands or even more, which is beyond the capacity of a stand-alone machine. %Although there are research works developing large-scale SDP solvers on multiple processors with distributed memory~\cite{Fujisawa:2012:HGS:2388996.2389122}, we focus on a single computing platform solution in this study. 
%on a dataset of limited size
One way to alleviate the problem is to factorize the bipartite graph into several independent subgraphs, which do not share common buyers or sellers. Then we can run the SDP solver on each subgraph individually and merge all results. This can be achieved considering the way we create the edges between buyers and sellers, i.e., each seller is connected to buyers in a sorted subsequence. In the following, we show results obtained on a small subgraph to avoid repeated comparison.

We create a subgraph of small size consisting of top $5$ sellers and top $26$ buyers (ranked by monetary). Each seller is connected with $10$ buyers so that each buyer can be assigned to multiple sellers. We use two types of edge weight, money and node rank. Money weight is computed in the same way as in the experiment of {\em C-REC}, i.e., the sum of the buyer's total monetary purchases (on all sellers within the category) and the seller's total profits (from all buyers within the category). Of course, any monotone weight function can be used. Rank weight equals the multiplication of a big constant (the total number of buyer and seller nodes in the entire graph, $20,626$ in our case) and the reciprocal of the sum of the buyer's rank and the seller's rank.

We create different number of conflict buyer pairs by randomly sampling the number of buyer pairs by the following percentages of total possible buyer pairs, $\{5\%, 10\%, 15\%, 20\%\}$. Similar to {\em C-REC}, we also set different degree constraint ratios for each node, $\{20\%, 30\%, 40\%, 50\%, 60\%\}$. In addition, we use a constant (``1") for each seller's conflict constraint, which means each seller can at most accept one conflict buyer pair. We use this constant to amplify the impact of conflict constraint on the small-scale subgraph. For SDP with rounding, we run the randomized rounding procedure $20$ times and take the best solution. 

Figure~\ref{fig:p3_money_solution} and Figure~\ref{fig:p3_rank_solution} depict the comparison of three approaches for the money weight and rank weight, respectively. 

In Figure~\ref{fig:p3_money_solution}, we observe that regardless of conflict pair ratio, solutions of the three methods are very similar when degree constraint ratio is smaller than $60\%$, with SDP with rounding being slightly weaker than others. The reason is that when the degree constraint is tight, the conflict constraint of each seller is less likely to be activated, i.e., chances are rare for multiple conflict buyers to be recommended to a seller. The difference arises when the degree constraint ratio and the conflict pair ratio are both weak and the higher conflict pair ratio results in larger performance drop for both SDP with rounding and GREEDY (e.g., comparing the sets of bars of $60\%$ degree constraint ratio in subfigures (c) and (d)). 

%For example, with $60\%$ degree constraint ratio and $20\%$ conflict pair ratio, SDP with rounding and GREEDY achieve $96.8\%$ and $88.8\%$ of the optimal.
% inferior

Figure~\ref{fig:p3_rank_solution} shows a similar performance change trend for the three approaches. For example, the solution difference becomes larger when degree constraints are weaker and the performance of SDP with rounding and GREEDY gradually decreases as the conflict pair ratio increases. Meanwhile, we also observe that the performance of GREEDY is always superior to that of SDP with rounding. \nop{We hypothesize that the edge weight also affects the solution of different methods.} 

Both SDP with rounding and GREEDY achieve close to optimal solutions. Specifically in the experiments, GREEDY exhibits a far superior performance compared to the theoretical analysis. In addition, the most important advantage of GREEDY is that it scales very well even when applied to large-scale datasets. We show its scalability in Figure~\ref{fig:p3_greedy}. We use the same $2.0\%$ density graphs of different size as in the C-REC experiment. The settings of degree constraint ratio, conflict pair ratio are $60\%$ and $10\%$, respectively. For each seller, the conflict threshold $t$ (refer to Section~\ref{P3}) is set to be $50\%$ \nop{???}of the total number of conflicting buyer pairs associated to the seller. 

\begin{figure}[htb]
\centering
\includegraphics[scale=0.23]{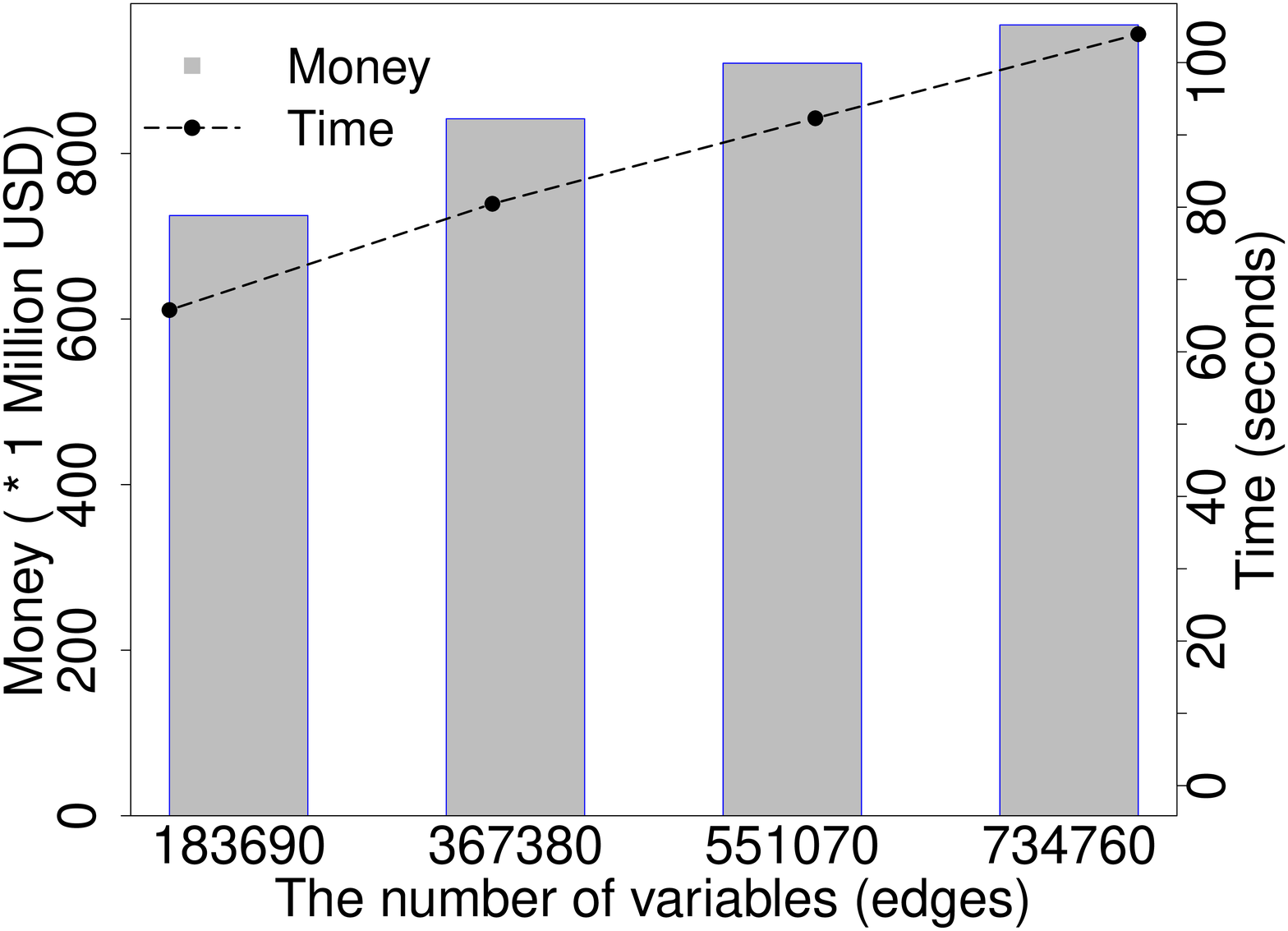}
\caption{Greedy algorithm on large-scale datasets showing its scalability.}
\label{fig:p3_greedy}
\end{figure}

In Figure~\ref{fig:p3_greedy}, the running time increases nearly linearly, and it only requires around $104$ seconds to get a solution when the number of variables (edges) is considerably large, $734,760$!

\subsubsection{ILP Formulation of CAC-REC}
ILP formulation enables us to take full advantage of the LP solver (Gurobi) to solve CAC-REC problems with larger sizes. In this section, we perform larger-scale experiments to compare solutions of different methods, i.e., ILP, LP with rounding and GREEDY.

We create a $0.16\%$-density bipartite graph using the same method described in Section~\ref{sec:c-rec_exp}. The full graph consists of $18742$ buyers, $1884$ sellers and $56520$ edges. We also extract three subsets of different sizes from the full graph, i.e., using $25\%$, $50\%$ and $75\%$ of the total number of edges (the number of buyer and seller nodes decreases accordingly). The degree constraint ratio, conflict pair ratio are $50\%$ and $10\%$, respectively. The conflict threshold $t$ (refer to Section~\ref{P3}) for each seller is set to be $50\%$ of the total number of conflicting buyer pairs associated to the seller. Figure~\ref{fig:lp_cac-rec} shows the solution comparison of different methods on different datasets.

\begin{figure}[htb]
\centering
\subfigure[Money Solution]{
\includegraphics[scale=0.2]{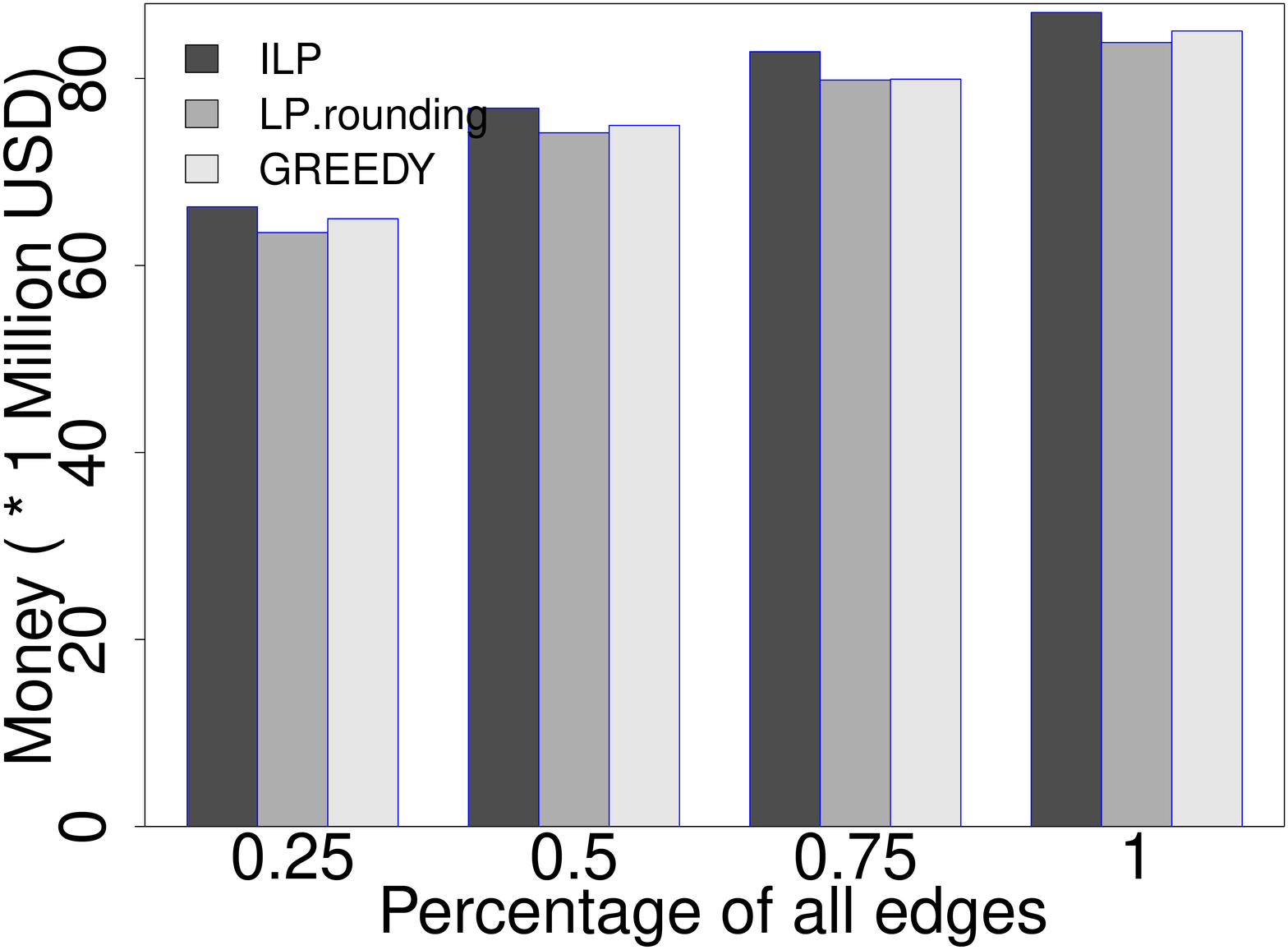}
}
\subfigure[Rank Solution]{
\includegraphics[scale=0.2]{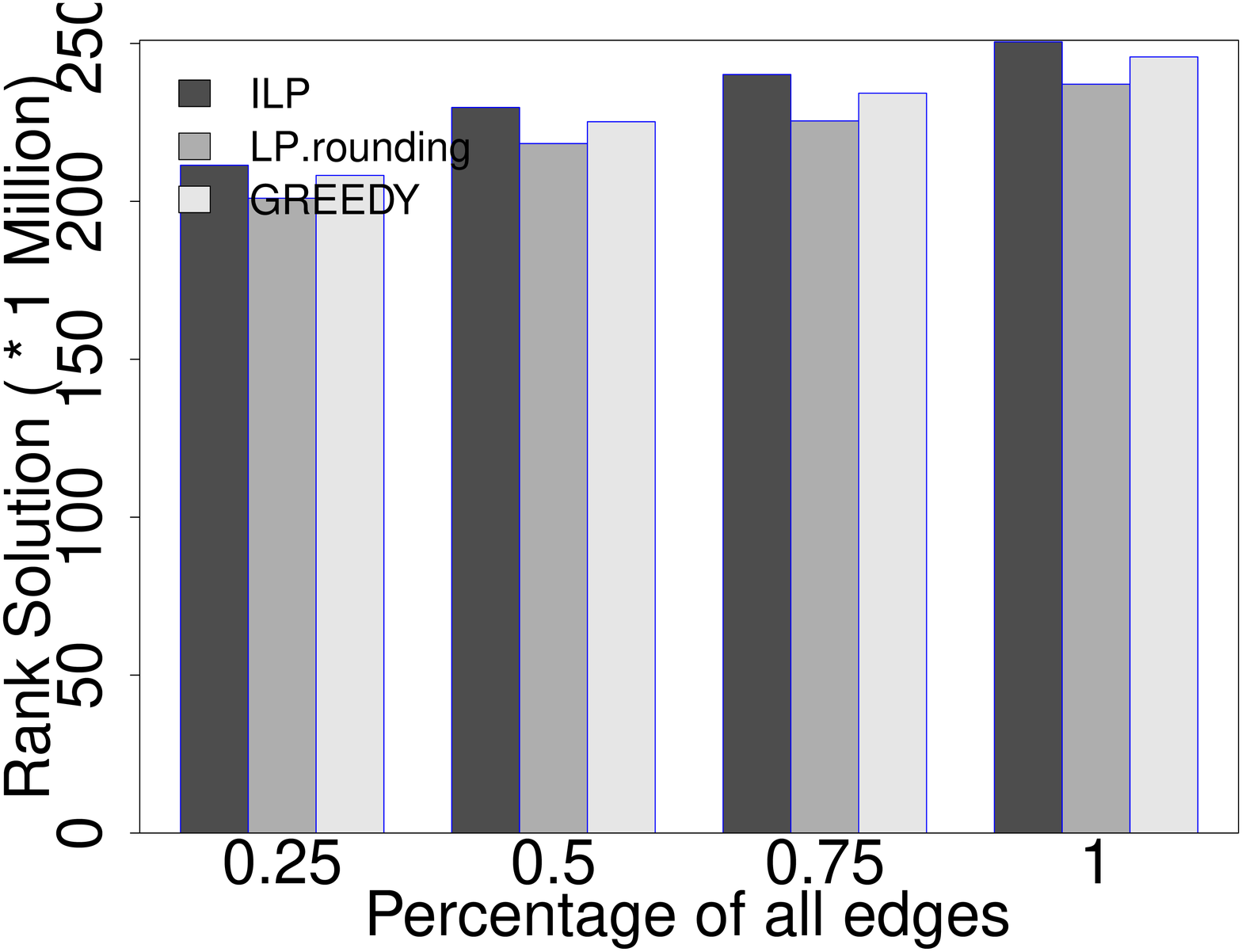}
}
\caption{Linear Program experiments of CAC-REC on large-scale datasets.}
\label{fig:lp_cac-rec}
\end{figure}

Figure~\ref{fig:lp_cac-rec} shows very promising results for LP with rounding and GREEDY algorithms on both money weights and rank weights; they are only slightly worse than the optimal integer solution obtained by ILP. Comparing to the SDP experiments, this experiment also verifies the effectiveness of both LP with rounding and GREEDY on large datasets because the graph we use in this section is considerably larger than the one used for SDP experiments. The density ($0.16\%$) is also $10$ times larger than the real-world graph ($0.016\%$). Therefore, our LP formulation improves the scalability of solving larger-scale CAC-REC problems. In addition, while ILP is NP-hard (running time may be long\footnote{In our experiment in this section, the longest running time of ILP is approximately $10$ minutes.}), we can use the much faster GREEDY method or LP with rounding to improve the efficiency. 

\nop{
I tried the LP approach using different densities. 
For 25% of a 0.5% density graph (0.1 conflict pair ratio), it failed and ran out of memory after several hours. The size of the problem is 775,894 (constraints) * 922,944 (variables, 753,384 of them are Z vars).

But it survived on a 0.16% density graph (0.1 conflict pair ratio), which is 10 times denser than the real-world dataset (0.016%) . Below is the running time for different sizes.

Percentage of edges || Number of edges ||  Size of the problem (constraints * num_vars (num_Z)) || Running time (seconds)
25% ||  14130  ||  26193 * 34680 (20550)   || 4.9438
50% ||  28260  ||  51837 * 68832 (40572)   || 29.9172
75% ||  42390  ||  77911 * 103413 (61023)  || 9.9947
100% || 56520  ||  104573 * 138583 (82063) || 56.9171

So the full bipartite graph with 56,520 edges only requires less than 1 minute. In SDP approach, we cannot make it even when the number of edges is only 100.
The LP formulation indeed improves the scalability of CAC-REC.
}

\section{Conclusions}
\label{conclusions}
We introduced a novel recommendation problem that aims at recommending buyers to sellers (RBS) under capacity and conflict constraints. We provided formal definitions of two types of RBS, C-REC and CAC-REC, addressing different RBS scenarios. We showed that C-REC could be effectively solved using linear programming. By considering the conflict between buyers, however, the complexity of RBS increases significantly. We proved that CAC-REC is NP-hard. Then we proposed a SDP algorithm with a rounding procedure and a greedy algorithm to solve CAC-REC. Our results of extensive experiments using real-world datasets demonstrated that the proposed algorithms can achieve close to optimal solutions. Finally, we showed that the greedy algorithm is highly scalable. %Future work involves continuing to investigate the RBS problem as well as related techniques to improve the scalability of the SDP algorithm.

%\end{document}  % This is where a 'short' article might terminate

%ACKNOWLEDGMENTS are optional
\section*{Acknowledgments}
We would like to thank Terapeak Inc. for providing the eBay data for this research.  

% The following two commands are all you need in the
% initial runs of your .tex file to
% produce the bibliography for the citations in your paper.
%\balance
\bibliographystyle{abbrv}

%\balance
\bibliography{bs}  % sigproc.bib is the name of the Bibliography in this case

% You must have a proper ".bib" file
%  and remember to run:
% latex bibtex latex latex
% to resolve all references
%
% ACM needs 'a single self-contained file'!
%
%APPENDICES are optional
%\balancecolumns

\nop{
\appendix
%\balancecolumns % GM June 2007
% That's all folks!
}

\end{document}